\documentclass[comsoc]{IEEEtran}
\IEEEoverridecommandlockouts

\usepackage{cite}
\usepackage{amsmath,amssymb,amsfonts}
\usepackage{graphicx}
\usepackage{textcomp}
\usepackage{xcolor}
\usepackage{epstopdf}
\usepackage{amsthm}
\usepackage{nicefrac}
\usepackage{balance}
\usepackage{tabularx,booktabs}
\usepackage{color}
\usepackage{subcaption}
\usepackage{mwe}
\usepackage{bbm}
\usepackage{algorithm} 
\usepackage{algpseudocode} 
\usepackage{subcaption}

\newtheorem{theorem}{Theorem}
\newtheorem{corollary}{Corollary}

\newtheorem{lemma}[theorem]{Lemma}
\newtheorem*{definition}{Definition}
\newtheorem{remark}{Remark}
\newtheorem{property}{Property}

\makeatletter
\if@twocolumn
\newcommand{\whencolumns}[2]{
	#2
}
\else
\newcommand{\whencolumns}[2]{
	#1
}
\fi
\makeatother

\begin{document}

\title{A Multi-Objective Optimization Framework for URLLC with Decoding Complexity Constraints
}

\author{Hasan~Basri~Celebi,~\IEEEmembership{Student Member,~IEEE,}
	Antonios~Pitarokoilis,~\IEEEmembership{Member,~IEEE,}
	and~Mikael~Skoglund,~\IEEEmembership{Fellow,~IEEE}
	\thanks{Hasan Basri Celebi and Mikael Skoglund are with the School of Electrical
		Engineering and Computer Science, KTH Royal Institute of Technology,
		Stockholm, Sweden (e-mail: hbcelebi@kth.se, skoglund@kth.se).}
	\thanks{Antonios Pitarokoilis is with Ericsson AB, Stockholm, Sweden (e-mail: antonios.pitarokoilis@ericsson.com).}
	\thanks{This work was funded in part by the Swedish Foundation for Strategic Research (SSF) under grant agreement RIT15-0091.}
}

\maketitle

\begin{abstract}
	Stringent constraints on both reliability and latency must be guaranteed in ultra-reliable low-latency communication (URLLC). To fulfill these constraints with computationally constrained receivers, such as low-budget IoT receivers, optimal transmission parameters need to be studied in detail. In this paper, we introduce a multi-objective optimization framework for the optimal design of URLLC in the presence of decoding complexity constraints. We consider transmission of short-blocklength codewords that are encoded with linear block encoders, transmitted over a binary-input AWGN channel, and finally decoded with order-statistics (OS) decoder. We investigate the optimal selection of a transmission rate and power pair, while satisfying the constraints. For this purpose, a multi-objective optimization problem (MOOP) is formulated. Based on the empirical model that accurately quantifies the trade-off between the performance of an OS decoder and its computational complexity, the MOOP is solved and the Pareto boundary is derived. In order to assess the overall performance among several Pareto-optimal transmission pairs, two scalarization methods are investigated. To exemplify the importance of the MOOP, a case study on a battery-powered communication system is provided. It is shown that, compared to the classical fixed rate-power transmissions, the MOOP provides the optimum usage of the battery and increases the energy efficiency of the communication system while maintaining the constraints.
\end{abstract}

\begin{IEEEkeywords}
	URLLC, low-complexity receivers, channel coding, internet-of-things, order statistics decoder, multi-objective optimization.
\end{IEEEkeywords}

\section{Introduction}

The advent of 5G communication technologies will provide more than mobile broadband data transmission by supporting massive number of machine-type connections and allowing communication with stringent latency and reliability constraints. With these two new communication frameworks, namely the massive machine-type communication (MTC) and ultra-reliable low-latency communication (URLLC), 5G will not only improve the human-centric connection, but also create a new interconnection network for devices, objects, sensors, such as the Internet of Things (IoT). However, compared to the previous generations, this new network will require significantly different architecture and operation modes. For instance, the target end-to-end radio latency limit for 4G networks is $10\,$ms \cite{imt_requirements}. However, this requirement is further reduced to $ 1 \,$ms for URLLC applications with very high reliability constraints, i.e. maximum error probability of $ 10^{-5} $ \cite{3gpp_release15}. Furthermore, MTC related URLLC scenarios are among the most challenging since in real time applications without human intervention, e.g., industrial control, extremely stringent requirements on latency, in the order of a fraction of ms, and reliability, error probability less that $ 10^{-9} $, must be met. \cite{lema_business, pokhrel_toward, celebi_wireless_comm}. 

Although the capacity of a channel is often taken as a metric for reliable communication, it is not a suitable metric for latency-intolerant applications, since it gives the ultimate error-free transmission rate when the blocklength of the transmitted codeword goes to infinity \cite{shannon_a_mathematical}. Therefore, the outage capacity, which represents the maximal transmission rate such that the probability of the instantaneous mutual information being lower than the selected rate is not higher than some desired threshold, is sometimes studied to evaluate the performance of a latency constrained communication \cite{popovski_wireless_access}. However, the outage capacity is more accurate for arbitrarily long blocklengths, which makes it suitable for latency-tolerant applications. Nevertheless, there has been a significant amount of progress made in the area of non-asyptotic bounds in finite blocklength regime. Non-asymptotical achievability and converse bounds are derived in \cite{polyanskiy_channel_coding}, where a close approximation on the maximal achievable rate is also introduced. This approximation reveals that a desired error probability in finite blocklength can be achieved by introducing some amount of rate penalty from the asymptotic bound, where the amount of the penalty is related to the channel dispersion and blocklength. 

\subsection{Motivation} 

The bounds in finite blocklength regime show the theoretical limits. However, achieving these limits is still a challenge. One significant determinant of this problem is the selection of a proper channel encoder and decoder pair which can achieve performance levels close to the theoretical limits. Therefore, significant amount of work has been published in the recent years to investigate the performances of various coding schemes in finite blocklength regime to find the optimum selection for URLLC applications. Mainly, a list of some strong candidates for URLLC  follows \cite{zaidi_5g_physical, shivarnimoghaddam_short_block} (\textit{i}) Bose–Chaudhuri–Hocquenghem (BCH) codes with order-statistics (OS) decoder \cite{fossorier_soft_decision}, (\textit{ii}) binary or non-binary low-density parity-check (LDPC) codes with OS decoder or belief-propagation algorithm \cite{saeedi_performance}, (\textit{iii}) tail-biting convolutional codes (TBCC) with list or wrap-around Viterbi algorithms \cite{gaudio_on_the}, and (\textit{iv}) polar codes with cyclic-redundancy-check aided successive-list decoding \cite{niu_crc}. Empirical performance results for these coding schemes have been investigated and compared in \cite{liva_code_design, shivarnimoghaddam_short_block, celebi_wireless_comm}. It is shown that linear block codes with OS decoder and TBCC with Viterbi algorithm outperform and their performance approach to the non-asymptotic limits. Due to their performance in finite blocklength regime, these two coding schemes have gained interest of the research community. Bounds on the performance of OS decoders are derived in \cite{yue_a_revisit} and a novel low-complexity algorithm is proposed. Performance of TBCC is studied in \cite{gaudio_on_the} for URLLC applications and lists of generator polynomials leading to large minimum distance convolutional codes are presented. On the other hand, structural delays, which is a significant factor for low-latency applications, for linear block codes and convolutional codes are studied in \cite{hehn_ldpc}. It is shown that convolutional codes have considerable advantages in terms of lower structural delays. Possible implementation opportunities of polar codes in 5G applications are discussed in \cite{bioglio_design}. Furthermore, research on deep-learning based coding schemes for short blocklengths is also attracting significant interest as it is shown in \cite{jiang_latency_performance, nachmani_deep} that it is possible to create deep-learning based coding schemes for URLLC applications. Additionally, learning based decoder complexity reductions for convolutional codes and OS decoders are discussed in \cite{ma_statistical} and \cite{cavarec_a_learning}, respectively. 

In general, performance of a decoder comes with an unavoidable cost: its computational complexity. The trade-off between the performance and complexity yields many studies to mainly focus on decreasing the computational cost of a decoder. Such trade-off, for instance, can be observed in OS decoders and TBCCs by changing their order and memory sizes, respectively. It is shown that such decoders with higher order/memory perform better. Besides, their complexity, in terms of number of binary operations per-information-bit, exponentially increases. This trade-off has been investigated in \cite{celebi_latency_and} extensively and it is shown that decoding complexity has a considerable effect in URLLC applications when processing capabilities are taken into account. In this work, we further extent \cite{celebi_latency_and} and formulate a multi-objective optimization framework for URLLC applications with the goal of maximizing the throughput, in terms of transmitted information rate, and energy efficiency of the system together. Such a goal is crucial for the design of URLLC IoT use cases where reliable information transmission is required under latency constraint with an energy efficient communication protocols, since very long battery lifetimes are required \cite{azari_serving}. Thus, we identify the optimal selections of communication parameters  where constraints on both latency and reliability are met under decoding complexity constraint. Some of the similar works in the literature, where the optimal parameters for URLLC are studied, can be listed as follows. Optimum rate and power allocation for URLLC are investigated in \cite{lopez_joint_power}. This work has been extended to non-orthogonal multiple access networks to achieve a distributed rate control system design with reliability constraints \cite{lopez_distributed}. The trade-off between energy efficiency and reliability is studied in \cite{zogovic_energy} for low-power short-rage communication, where a multi-objective optimization problem (MOOP) has been formulated and achievability of the optimal parameters are studied. Similarly, optimum power allocations under reliability constraints for hybrid automatic repeat request schemes are presented in \cite{dosti_ultra}. However, no constraint on decoding complexity has been considered in these papers.

\subsection{Contributions}

This work differs considerably from the studies listed above as we take into account the decoding complexity as a system design parameter for URLLC applications. Such an additional parameter is relevant for most of the real-world applications and it is mostly the case for the URLLC IoT use cases \cite{5g_americas, celebi_low_latency, lopez_enabling}. The main contributions of this work can be mainly summarized as 
\begin{itemize}
	\item We extend our analysis in \cite{celebi_latency_and} where the trade-off between the performance of several decoders, that are possible candidates of URLLC, and their computational complexity is modeled. With the help of this model, in this paper, we show that under decoding complexity constraint on OS decoders, a back-off from the finite maximal achievable rate is required in order to meet the latency and reliability constraints. We quantify this back-off both in power and rate domains. Then, based on this back-off, the optimum transmission parameter choices are discussed thoroughly via a MOOP framework.
	\item Next, a MOOP for the optimum rate and power selection under latency, reliability, and decoding complexity constraints is formulated. The MOOP is then analytically solved and the attainable objective set and the set of optimal selections, namely the Pareto boundary, are derived.
	\item Two methods for scalarization of the MOOP, namely (\textit{i}) linear weighted-sum and (\textit{ii}) weighted Chebyshev objective functions, are introduced and their performances are compared. It is subsequently shown that based on the shape of the attainable objective set, some optimal points are not accessible, depending on the selected scalarization method. This phenomena is discussed thoroughly and we analytically derive the regions where the optimum points are not accessible.
	\item Finally, the importance of the MOOP is exhibited with a case study on battery-powered communication. It is shown that the MOOP increases the energy efficiency of the communication system while maintaining the constraints, compared to the classical fixed parameter transmissions. We also show that weighted Chebyshev objective function performs better than linear weighted-sum function in terms of energy efficiency.
\end{itemize}
 
Next, we discuss the system model and formulate the problem in Section II. The trade-off between excess power and decoding complexity is investigated for OS decoders in Section III. Then, in Section IV, we formulate a MOOP and solve it analytically. Two main scalarization techniques for the MOOP are introduced in Section V. Finally, the benefits of the MOOP is shown in a case study, where performance of a battery-powered communication is presented. 

\section{Problem Definition}\label{sec_system_model}

\subsection{System Model}

We consider communication over a discrete-time, binary-input AWGN (BI-AWGN) channel\footnote{We study a communication scenario that can be generalized to practical linear multidimensional modulation formats. The choice of the BI-AWGN channel does not restrict the generality of the study.}. $ k $ number of information bits are encoded with a linear block encoder and modulated to form the  $ n $ symbol codeword
\begin{equation}\label{key}
\boldsymbol{x} = [ x_1, x_2, \ldots, x_n ], ~~ x_i \in \{-1,+1\}, 
\end{equation}
The codeword $ \boldsymbol{x} $ is then transmitted over BI-AWGN channel and the ratio $ r=k/n $ is the code rate of the selected codebook. The observed sequence at the receiver is 
\begin{align}\label{eq_system_model}
\boldsymbol{y} = \sqrt{\rho}\boldsymbol{x} + \boldsymbol{z},
\end{align}
where $\rho$ denotes the signal-to-noise ratio (SNR) and $\boldsymbol{z}\sim\mathcal{N}(\boldsymbol{0},\boldsymbol{I}_n)$. 

We consider a communication scenario where data transmission is performed under latency, reliability, and decoding complexity constraints. The transmitter and receiver select a transmission pair, denoted as $ [r,\rho] $, according to the constraints. The receiver captures $ \boldsymbol{y} $ and selects a decoder that fulfills the constraints. The goal of this study is to investigate the optimum choice of transmission parameters that guarantees to fulfill the constraints and maximize the efficiency of the communication system in terms of the objective function, which will be discussed in the following Sections. 

Information theoretic analysis tells that there exists a maximal limit on transmission rate where reliable communication is possible, such that the codeword error probability (CEP) vanishes as $ n \rightarrow \infty $. This limit is termed as the channel capacity and denoted as $ C $. However, if the communication setup is restricted to have arbitrarily large values of $ n $, such restrictions are common for URLLC applications with stringent latency requirements, $ C $ overestimates the rate of reliable information and is not achievable. In this case the maximal transmission rate with codewords of length $ n $ with non-zero CEP $ \varepsilon $ can be closely approximated by \cite{polyanskiy_channel_coding}
\begin{equation}\label{eq_normal_approximation}
R(n,\rho,\varepsilon)=C-\sqrt{\frac{V}{n}}Q^{-1}(\varepsilon)\log e+O\left(\frac{\log n}{n}\right).
\end{equation}
where the quantity $ V $ is called the channel dispersion and $ Q^{-1}(\cdot) $ is the inverse of the Gaussian $ Q- $function.\footnote{All logarithms in this paper are with base 2.} The expression in \eqref{eq_normal_approximation} is termed as the normal approximation to the maximal transmission rate in finite blocklength regime. For BI-AWGN channel, $ C $ and $ V $ are
\begin{align}\label{eq_channel_capacity}
C &= \frac{1}{\sqrt{2\pi}} \int e^{-\frac{z^2}{2}}  \left( 1-\log\left( 1+e^{-2\rho+2z\sqrt{\rho}} \right)  \right) \mathrm{d}z,
\\
V &= \frac{1}{\sqrt{2\pi}} \int e^{-\frac{z^2}{2}}  \Big( 1 - \log\left( 1+e^{-2\rho+2z\sqrt{\rho}} \right)  - C \Big)^2 \mathrm{d}z.
\end{align}
The second term in \eqref{eq_normal_approximation} introduces a back-off from $ C $ to ensure the transmission achieves $ \varepsilon $ CEP with $ n $ blocklength. Hence, in order to meet the constraint on reliability, one can deduce that $ r \leq R(n,\rho,\varepsilon) $. However, this constraint is not enough to ensure that the latency constraint is fulfilled already under decoding complexity constraint, since, as will be shown throughout the paper, the computational complexity of the decoder exponentially increases as $ r $ approaches to $ R(n,\rho,\varepsilon) $.

\subsection{Aggregate Latency and Decoding Complexity} \label{sec_dec_comp_prop}

For notational consistency, given a fixed codebook containing $ 2^{k} $ codewords of length $ n $, we denote a decoder as $ D ( n, r, \rho) $, where it is meant that the decoder operates on the given codebook at a received SNR, $ \rho $, and decodes $ k $ number of information bits. The performance of the decoder is measured by its CEP, denoted as $ \varepsilon $. 

It is often encountered in the literature of latency-constrained communication that the aggregate latency equals to the time required for the transmission of a message over the communication channel. This implies that other operations necessary for the successful delivery of the information, e.g., the time required for the decoding, are assumed to happen instantaneously. However, this assumption is not true for receivers with computational complexity constraints. Therefore, we consider the aggregate latency $ L_a $ as   
\begin{equation}\label{eq_latency_constraint}
L_a = nT_s + L_d + L_q ,
\end{equation}
where $T_s$ is the symbol time and $ L_d$ and $ L_q $ denote the decoding latency and the duration for signal propagation and other auxiliary processes such as signal processing and consequent decision making tasks, respectively. Since $ L_q $ does not depend on encoding and decoding processes, we neglect it from further analysis. 

For simplicity and generality, a linear relation between $ L_d $ and the total number of binary operations of the decoding process  is assumed \cite{hwang_advanced_computer, celebi_latency_and}.  Denoting $T_b$ as the latency caused by a single binary operation at the receiver, the total transmission and decoding latency for the transmission of a codeword of blocklength $n$ can be written as\footnote{We assume that decoding starts right after the codework transmission.}
\begin{equation}
L_a = n T_s +  k T_b K(D) ,
\label{eq_total_latency_with_c}
\end{equation}
where $ K(D) $ represents the number of binary operations per-information-bit that is required to decode $ k $ number of information bits from an $ n $-length noisy codeword. Notice that $ kK(D) $ gives the total number of binary operations required for decoding.\footnote{A more accurate estimation on $ L_d $ can be done by investigating software optimization capabilities, memory timings, parallel computation, etc. But since these are not in the scope of this paper, we confine \eqref{eq_total_latency_with_c} for further analysis. Interested readers may refer to \cite{wilhelm_the_worstcase}}

Next, in order to address the relation of $ K(D) $ with other communication parameters, we define some easily verifiable properties of the considered codes by comparing the relative performance of two decoders from the same decoder family operating on the same codebook or on the sub-codebooks \cite{celebi_latency_and,sybis_channel_coding,savage_complexity_of}.
\begin{property}
	Let two decoders, $ {D}_1(n,r,\rho) $ and $ {D}_2(n,r,\rho) $ with CEP $ \varepsilon_1 $ and $ \varepsilon_2 $, respectively, operating on the same codebook with $ {K}({D}_1) \leq {K}({D}_2) $. It follows immediately from the decoder complexities that $ \varepsilon_1 \geq \varepsilon_2 $. Intuitively,  more complex decoder leads to lower CEP.
\end{property}
\begin{property}
	Let two decoders,  $ {D}_1(n,r,\rho_1) $ and $ {D}_2(n,r,\rho_2) $, operating on the same codebook with same complexity, $ {K}({D}_1) = {K}({D}_2) $, but different SNR levels such that $ \rho_1 \leq \rho_2 $. Then, it must be true that $ \varepsilon_1 \geq \varepsilon_2 $, since higher SNR leads to lower CEP.
\end{property}
\begin{property}
	Take the following two decoders, $ {D}_1(n,r_1,\rho) $ and $ {D}_2(n,r_2,\rho) $. Assume that $ {D}_1(n,r_1,\rho) $ is operating over a sub-codebook of $ {D}_2(n,r_2,\rho) $, where $ r_1 \leq r_2 $, at the same SNR and complexity, i.e. $ {K}({D}_1) = {K}({D}_2) $. It is true that due to the size of the sub-codebook, $ \varepsilon_1 \leq \varepsilon_2 $. Intuitively, more information leads to higher CEP.
\end{property}
\begin{property}
	For the two decoders stated in Property 3, $ {D}_1(n,r_1,\rho) $ and $ {D}_2(n,r_2,\rho) $, where $ r_1 \leq r_2 $, $ \varepsilon_1 = \varepsilon_2 $ can be achieved when $ {K}({D}_1) \leq {K}({D}_2) $. Thus, more information leads to higher complexity.
\end{property}

\subsection{Power Gap and Rate Gap}

Suppose that, given $ n $ and $ \varepsilon $, based on the SNR $ \rho $, the transmitter and receiver agree on a transmission rate that meets the CEP constraint. We denote this rate-power transmission pair as $ [r,\rho] $. The transmission efficiency, in terms of increasing information transmission per channel use, is maximized by selecting the codebook that can achieve $ r = R(n,\rho,\varepsilon) $. However, as shown in \cite{liva_code_design, shivarnimoghaddam_short_block, celebi_latency_and}, this selection can lead to a very complex decoder which is not practical for receivers having complexity constraints since it may violate the latency constraint or it will make the total latency too large. For instance, the optimum maximum likelihood decoder, which requires an exhaustive search over the codebook, implies exponential complexity in $ k $ which is not practical even for short blocklengths.

Suppose that the total latency $ L_a $ is bounded by a maximum allowed latency $ L_m $ as
\begin{equation}\label{eq_latency_deadline}
L_a \leq L_m ,
\end{equation}
where $ L_m $ represents the maximum allowed latency. This constraint imposes an upper bound on the per-information-bit decoder complexity such that
\begin{equation}\label{eq_upper_bound_on_K}
{K}({D}) \leq \kappa ,
\end{equation}
where
\begin{equation}\label{key}
\kappa = (kT_b)^{-1}[L_m-nT_s]^+ 
\end{equation}
and $ [z]^+ = \max\{0,z\} $. However, an upper bound on $ K(D) $ for the complexity of the decoder for fixed $ n $, $ r $ and $ \rho $ would inevitably lead to reduced reliability, i.e., increased CEP. One way to satisfy a desired target reliability is to spend some amount of excess power, named as the power penalty or to introduce some amount of rate back-off. Hence, an interesting, yet complex, relation between power, rate, aggregate latency, and decoding complexity arises. 

Next, we introduce two definitions that will be used in further analysis.
\begin{definition}[Power penalty, $ \Delta \rho $]
	Fix a codebook of $ 2^k $ codewords of blocklength $ n $. For a reference SNR, $ \rho_s $, consider the ML decoder that achieves a CEP of $ \varepsilon^* $ and the suboptimal decoder $ {D}(n,r,\rho) $ that achieves $ \varepsilon^* $ at SNR $ \rho $. The difference between $ \rho $  and $  \rho_s $ is the power penalty required, such that the suboptimal decoder can achieve the same CEP as the ML decoder.
\end{definition}

\begin{definition}[Rate gap, $ \Delta r $]
	Consider the ML decoder, operating on a codebook with $ 2^k $ codewords of blocklenth $ n $ at a coderate $ r $, achieves $ \varepsilon^* $ CEP at the reference SNR $ \rho_s $. It is true that one can further decrease the CEP to $ \varepsilon' $, where $ \varepsilon' < \varepsilon^* $, by omitting sufficient amount of codewords from the codebook and let the ML decoder to operate over a sub-codebook. Although this operation reduces the coderate to $ r' $, where $ r' < r $ , it also gives the flexibility of selecting a suboptimal decoder, $ {D}(n,r',\rho_s) $ with  $ \varepsilon^* $ but substantially lower decoding complexity that can decode $ k' = nr' $ number of information bits at SNR $ \rho_s $. Thus, the difference between $ r $  and $  r' $ is the rate gap, such that the suboptimal decoder can achieve the same CEP with the same SNR as the ML decoder by sacrificing some amount of coderate.
\end{definition}

\subsection{Problem Formulation}

Several optimization problems for URLLC with decoding complexity constraints have been introduced and solved in \cite{celebi_latency_and}. The main logic in all those optimization problems is minimizing or maximizing a single cost function of interest subject to a set of constraints. In those single-objective optimization problems there are two main assumptions such as (\textit{i}) the objective function has dominating importance, (\textit{ii}) a-priori information about the good values for the constraints is already known. However, in real life implementations of URLLC, these assumptions may not hold and various parameters are supposed to be optimized together. Here, we take our analysis in \cite{celebi_latency_and} one step further and set the optimal design of URLLC systems in a MOOP framework.

A general structure of a MOOP follows 
\begin{subequations}\label{eq_gen_moop}
	\begin{align}
	\underset{x}{\text{minimize}} & ~~ [f_1(x), ~ f_2(x), ~ \cdots, ~ f_m(x)] \label{eq_gen_moop_obj}
	\\
	\text{s.t.} & ~~ x \in X,
	\end{align}
\end{subequations}
where $X$ represents the set of available resources and the goal of the objective in \eqref{eq_gen_moop_obj} is to minimize all the $m$ number of objectives simultaneously \cite{bjornson_multiobjective,lei_overview}. Similar to \eqref{eq_gen_moop}, possible objectives for MOOPs for URLLC can be formulated by a valid combination of the following constraints, such as (\textit{i}) minimization of $ L_a $, (\textit{ii}) minimization of $ \varepsilon $, (\textit{iii}) minimization of $ \rho $, (\textit{iv}) minimization of $ K(D) $, and (\textit{v}) maximization of $ r $, where the optimization performs subject to constraints on the remaining parameters. Since maximization of $r$ and minimization of $\rho$ are one of the two key parameters of a communication system, in this paper, we select $m=2$ and focus on these two objectives. 


Suppose an $ n $-blocklength codeword that belongs to a codebook of size $ 2^{nr_s} $, where $ r_s $ stands for the reference transmission rate, is intended to be transmitted at a CEP $ \varepsilon $. If no constraints on latency and decoding complexity present or if $ T_b = 0 \, $s, which stands for infinite computation power, the reference transmission rate-power is $ [ r_s, \rho_s ] $,
where
\begin{equation}\label{eq_reference_power}
\rho_s = R^{-1}(n,r_s,\varepsilon) 
\end{equation}
represents the reference SNR. However, constraints may prevent to achieve $ [r_s, \rho_s ] $. Suppose, for instance, a selected transmission rate-power pair that meets the constraints is 
\begin{equation}\label{key}
[ r_s - \Delta r, ~ \rho_s + \Delta \rho ] .
\end{equation} 
This selection arises a very significant optimization problem, being ``\textit{Given a blocklength $ n $ and reference transmission rate-power pair, $ [ r_s, \rho_s ] $, under latency, reliability, and decoding complexity constraints, what is the optimum selection of power penalty and rate gap that allows the transmission to satisfy the constraints under some optimality criterion?}". This MOOP can be formulated as the following

\begingroup
\allowdisplaybreaks
\begin{subequations}
	\label{eq_mul_opt_prob}
	\begin{align}
	\underset{k,\varepsilon}{\text{minimize}} & ~~ [\Delta r, ~ \Delta \rho]   \label{eq_mul_opt_problem_1_obj}
	\\
	\text{s.t.}  & ~~ \varepsilon \leq \varepsilon_m , \label{eq_mul_opt_problem_1_const_1}
	\\
	& ~~ K(D) \leq \kappa, \label{eq_mul_opt_problem_1_const_2}
	\\
	& ~~ 0 \leq \Delta \rho, \label{eq_mul_opt_problem_1_const_3}
	\\
	& ~~ 0 \leq \Delta r \leq r_s, \label{eq_mul_opt_problem_1_const_4}
	\\
	& ~~ k \leq n,  \label{eq_mul_opt_problem_1_const_5}
	\end{align}
\end{subequations}
\endgroup
where \eqref{eq_mul_opt_problem_1_const_1} and \eqref{eq_mul_opt_problem_1_const_2} represent the reliability and latency constraints, respectively, and $ \varepsilon_m $ is the maximum allowed CEP of the decoder. \eqref{eq_mul_opt_problem_1_const_3} and \eqref{eq_mul_opt_problem_1_const_4} are the numerical constraints on the variables. Since the hardware platform is assumed to be fixed in \eqref{eq_mul_opt_prob}, $ T_b $ and $ T_s $ are fixed. Notice that the overall goal of this MOOP is to achieve the ultimate point $[0,0]$. However, although selecting $ \Delta \rho = 0 $ and $ \Delta r = 0 $ is theoretically achievable, it requires very complex decoder to achieve the desired CEP. On the other hand, selecting $ \Delta \rho > 0 $ and $ \Delta r > 0 $ yields a reduction in decoder complexity, but it will also cause performance degradation in the transmission efficiency since the rate-power pair is receding from the ultimate point $[0,0]$.

\section{Power Penalty vs Decoding Complexity} 

In order to address the MOOP in \eqref{eq_mul_opt_prob}, the relation between the power gap versus the decoding complexity must be studied thoroughly. From now on, we focus our analysis on linear block codes with OS decoders, due to the following reasons; \textit{i}) it has been shown that these codes come very close to the information-theoretic bounds for finite $ n $ \cite{liva_code_design,zaidi_5g_physical}, \textit{ii}) OS decoders are easy to describe with a few parameters., \textit{iii}) their decoding performance can be easily parameterized by a single parameter, i.e., the order of the decoder, $ s \in \mathbb{Q} $, and finally, \textit{iv}) operations that are executed during decoding can be accurately tracked and the decoding complexity can be intuitively described.

\subsection{OS Decoders}

OS decoders are universal soft-decision decoder for linear block decoders and they are efficient in reducing the complexity of the decoding process by making soft decisions on a set of test codewords. Implementation of the OS decoders can be summarized as:
\begin{itemize}
	\item Find a permutation function $ \lambda(\boldsymbol{y}) $ that sorts the received vector, $ \boldsymbol{y} $, with respect to the absolute amplitude values.
	
	\item Using the permutation function $ \lambda(\cdot) $, reorder the columns of the generator matrix, $ \boldsymbol{G} $, and apply the Gauss-Jordan elimination to form the new systematic generator matrix $ \boldsymbol{G}_\lambda $.
	
	\item Set the list of test error patterns (TEPs), denoted by $ T $, such that it includes all possible length$ -k $ binary sequences with Hamming distance less than or equal to $ s $ and search over the list to find the error sequence that maximizes the likelihood of the codeword to the hard decoded $ \lambda(\boldsymbol{y}) $. 
\end{itemize} 
Overall, the decoding can be formulated as
\begin{align}
\boldsymbol{r}^s &= \underset{ \left\lbrace \boldsymbol{r}: \: \boldsymbol{r} = (h(\lambda(\boldsymbol{y})) \oplus \boldsymbol{t}) \otimes \textbf{G}_\lambda, \: \boldsymbol{t}\in T \right\rbrace }{\text{arg max}}  \mathbb{P}(\boldsymbol{r}|h(\lambda(\boldsymbol{y})) ,
\\
&=\underset{ \left\lbrace \boldsymbol{r}: \: \boldsymbol{r} = (h(\lambda(\boldsymbol{y})) \oplus \boldsymbol{t}) \otimes \textbf{G}_\lambda, \: \boldsymbol{t}\in T \right\rbrace }{\text{arg min}}  \| \boldsymbol{r} - h(\lambda(\boldsymbol{y})) \|_2 ,
\label{eq:MLDecoder}
\end{align}
where $h(\cdot)$ represents the hard-decoding function. The output of the order$ -s $ decoder is $ \hat{\boldsymbol{r}}^s=\lambda^{-1}\left( \boldsymbol{r}^s \right) $.

\subsection{Computational Complexity}

The search space for the most probable codeword is controlled by limiting the size of the list $ T $, which is related with the choice of the order $ s $. The total number of possible binary error vectors in TEP for an order$ -s $ OS decoder, when $ s $ is integer, is
\begin{equation}\label{key}
|T| = \sum_{i=0}^{s} {k\choose i}. 
\end{equation}
Therefore, the cardinality of $ T $ grows exponentially in $ s $. Notice that all possible codeword comparisons will be performed when $ s = k $ and the performance and complexity of the OS decoder will be identical to the ML decoder. The number of binary operations per-information-bit of an OS decoder can be calculated by \cite{celebi_latency_and}
\begin{equation}\label{eq_complexity_of_os_dec}
K( D ) = \frac{\log(n)}{r} + nk + \frac{1}{2} |T| \left( n - q + \frac{qn}{k} \right) ,
\end{equation}
where $ q $ represents the number of quantization bits used to represent the symbols of $ \boldsymbol{y} $ a real number. The terms in \eqref{eq_complexity_of_os_dec} represent the binary complexity values of sorting $ \boldsymbol{y} $, Gauss-Jordan elimination of the permuted $ \boldsymbol{G} $, and order$ -s $ reprocessing, respectively. Notice that the limiting complexity order of OS decoder is $ \mathcal{O}(nk^s) $.

Next, we validate the properties listed in Sec. \ref{sec_dec_comp_prop} by analyzing the empirical performance of OS decoders with various orders and transmission rates. 
In Fig. \ref{fig_CER_n_128_k_64_71}, performance results of OS decoders with orders $ s=\{0,1,2\} $ where $ n=128 $ and $ k=\{64, 71\} $ are depicted with dotted lines. The extended BCH code \cite{reed_error_control} is used for the encoding. The error bounds, shown with solid lines, are calculated from \eqref{eq_normal_approximation}. Property 1 and 2 can be seen by fixing $ k $ and letting the order$ -s $ to vary. It is seen that as the order increases, the performance of the decoder improves and approaches the optimal decoder, albeit at the expense of higher decoding complexity. Similar results are also achieved by fixing $ s $ and let SNR vary.  Property 3 can be illustrated by fixing an order and letting $ k $ to vary. It can be seen that for fixed SNR, the decoder with $ k=71 $ has worse performance than the decoder with $ k=64 $ and, from \eqref{eq_complexity_of_os_dec}, it also has higher per-information-bit complexity.

\begin{figure}[t]
	\centering
	\whencolumns{
		\includegraphics[width=.6\linewidth]{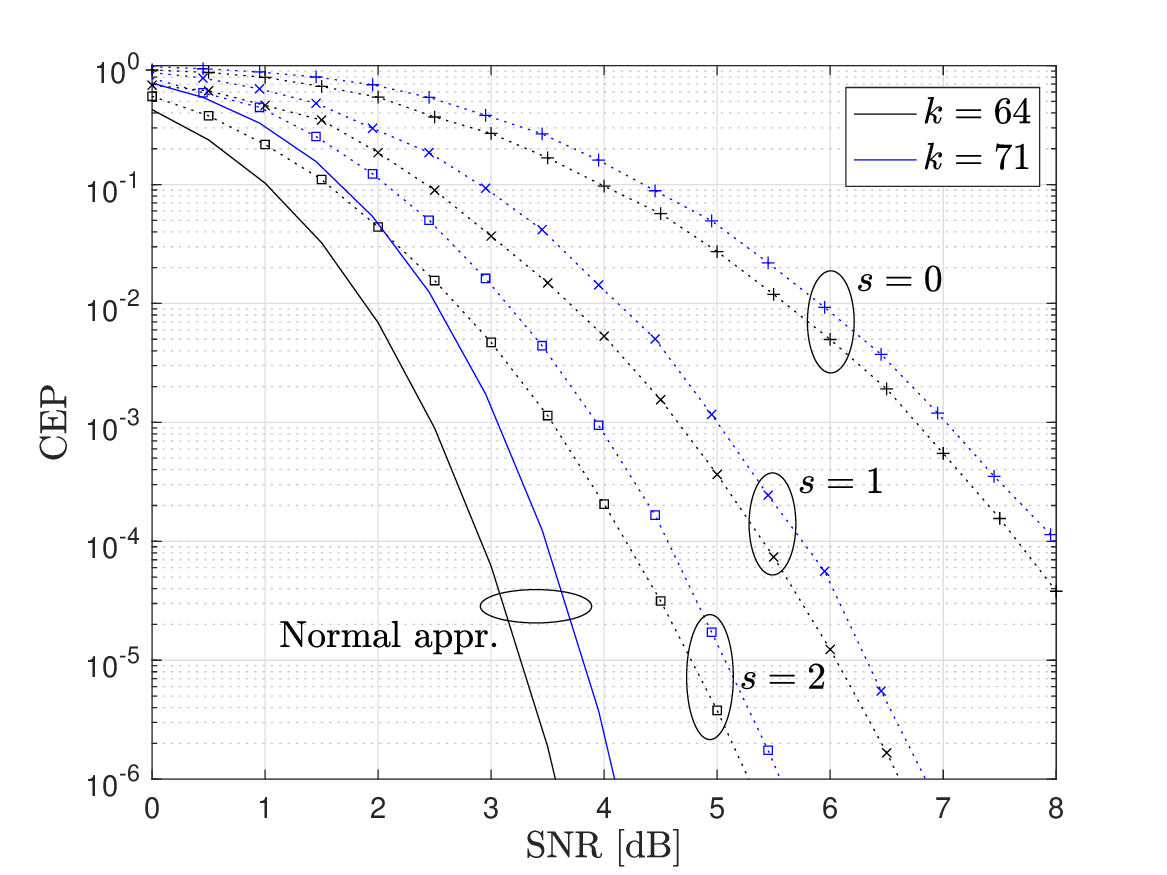}
	}{
		\includegraphics[width=1\linewidth]{figures/CER_n_128_k_64_71.eps}
	}
	\caption{CEP performace of OS decoders with $ n=128 $ and $ k=64 $ (black) and $ k=71 $ (blue) number of information bits at different orders compared to the the normal approximation CEP bound for BI-AWGN channel.}
	\label{fig_CER_n_128_k_64_71}
\end{figure}

\subsection{Modeling the Power Penalty}
\label{sec_model_power_penalty}

One can obtain $ \Delta {\rho} $ by simply calculating the difference between $ \rho' $ and $ \rho_s $, where $ \rho' $ is the the minimum SNR that is required for a particular OS decoder to achieve the target CEP, $ \varepsilon_m $. Thus, $ \rho' $ can be computed by solving the following optimization problem
\begin{equation}
\rho' = \underset{\{\rho \in \mathbb{R^+}, ~ \varepsilon \leq \varepsilon_m \}}{\mathrm{min}} \rho. 
\end{equation} 
Although tight upper bounds on $ \rho' $ have been derived in \cite{fossorier_soft_decision} and \cite{dhakal_on_the}, computation of $ \Delta {\rho} $ is not trivial for OS decoders since no closed-form expression of $ \rho' $ is available. 

On the other hand, it is shown in \cite{liva_code_design, shivarnimoghaddam_short_block, lian_performance, celebi_low_latency, maiya_coding_with} that for several types of coding schemes, such as linear block codes, polar codes, convolutional codes, etc., the relation between computational complexity and power penalty for fixed $ n $ subject to reliability constraint in the short block-length regime follows an exponential relation, where the decoding complexity exponentially increases as the code approaches the maximal achievable bound in  \eqref{eq_normal_approximation}. This behavior of computational complexity of the OS decoder and its power gap has been modeled in \cite{celebi_low_latency} and \cite{celebi_latency_and} by a law of the type
\begin{equation}
F(\Delta \rho) = \left( a \sqrt{\Delta \rho} + b \right)^{-1},
\label{eq_q_delta_tradeoff}
\end{equation}
with appropriate choices of the constants $a > 0 $ and $b > 0 $, where $ F(\Delta \rho) $ gives a close approximation of the logarithm of per-information-bit complexity of OS decoder that achieves the desired CEP. The main advantage of \eqref{eq_q_delta_tradeoff} is that it describes a tractable way of the trade-off between decoding complexity and power penalty for practical OS decoders with linear block codes. With the help of this model, given $ \Delta \rho $, it is possible to get a close approximation of the minimum per-information-bit complexity of the OS decoder that meets the desired CEP. A realization of $ F(\Delta \rho) $ is depicted in Fig. \ref{fig_power_gap_k_64_71}, where logarithm of per-information-bit complexities of the two OS decoders with orders $ s=\{0,1,2,3,4,5\} $ are shown with respect to their power gap values. It can be seen that the proposed model with $ a=0.029 $ and $ b=0.03 $ can closely describe the trade-off.

\begin{figure}[t]
	\centering
	\whencolumns{
		\includegraphics[width=.6\linewidth]{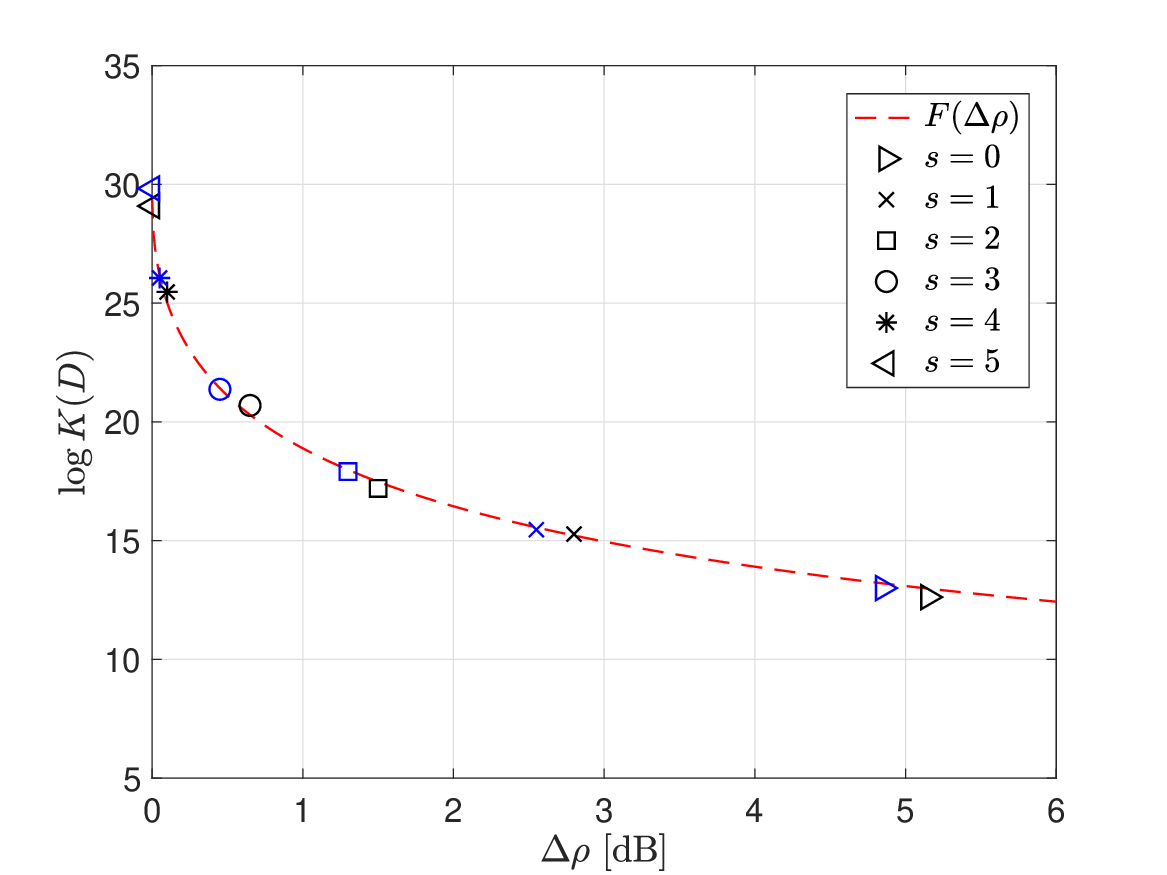}
	}{
		\includegraphics[width=1\linewidth]{figures/power_gap_k_64_71.eps}
	}
	\caption{Power penalty values of OS decoders at different orders versus their complexities for $ k=64 $ (black markers) and $ k=71 $ (blue markers), where $ n=128 $ and $ \varepsilon=10^{-5} $.}
	\label{fig_power_gap_k_64_71}
\end{figure}

Next, we show the maximal achievable information rate for a complexity constrained receiver with OS decoder under latency and reliability constraints.
\begin{lemma} \label{lemma_max_inf_rate}
	For a complexity constrained receiver with OS decoder and aggregate latency, expressed in \eqref{eq_total_latency_with_c}, the maximal achievable information rate subject to latency, $L_a<L_m$, and reliability constraints can be expressed as 
	\begin{equation}\label{eq_maximal_rate_with_latency}
	M(n,\rho, \varepsilon) = R(n, \rho - \Delta \rho^{\min}, \varepsilon) ,
	\end{equation}
	where $ M(n,\rho, \varepsilon) $ represents the maximal information rate under latency and complexity constraints and $ \Delta \rho^{\min} $ is the minimum amount of excess power that is required to fulfill the constraints and can be computed as  
	\begin{equation}\label{eq_min_req_power_penalty}
	\Delta \rho^{\min} = \left( (a\log\kappa)^{-1}[1- b\log\kappa]^+  \right)^2 .
	\end{equation}
\end{lemma}
\begin{proof}
	Lemma \ref{lemma_max_inf_rate} can be proven in accordance with Lemma 1 and 2 in \cite{celebi_latency_and}. For the purpose of completeness of this paper, we repeat the proofs. For fixed rate and blocklength $n$ the maximum allowable decoding time can be calculated using \eqref{eq_total_latency_with_c}. This in turn yields an upper bound on $ K(D) $ as given in \eqref{eq_upper_bound_on_K}. Bounding the complexity restricts the order$ -s $ as follows
	\begin{equation}
	s \leq s_m = \underset{\{s| s \in \mathbb{Q^+}, ~ L_a \leq L_m \}}{\mathrm{arg~max}} {K}({D}),
	\label{eq_general_bound_on_s}
	\end{equation}
	where $ s_m $ represents the maximum allowed decoder order$ -s $ that satisfies the decoding complexity. Although this restriction can be used to control the latency of decoding duration, the expense would be the reduced reliability, which can be satisfied by paying some amount of excess power. Thus, using \eqref{eq_q_delta_tradeoff} and \eqref{eq_upper_bound_on_K} the minimum amount of required power penalty can be computed as given in \eqref{eq_min_req_power_penalty}. Finally,  $ M(n,\rho, \varepsilon) $ can be determined by calculating $ \Delta \rho^{\min} $ for $ r \in [0,1] $ and shifting $ R(n, \rho, \varepsilon) $ by $ \Delta \rho^{\min} $ to the right.
\end{proof}

\begin{remark} \label{remark_power_gap}
Notice that $ \Delta \rho^{\min} $ is proportional to the processor capabilities, the latency requirements, blocklength $ n $, and coderate $ r $. For fixed $ n $ and $ T_s $, as $ T_b $ decreases, i.e. more powerful processor is implemented at the receiver, $ \Delta \rho^{\min} $ decreases and the gap to the normal approximation shrinks and disappears when  
\begin{equation}\label{key}
T_b \leq \frac{1}{k}2^{1/b}[L_m-nT_s]^+ .
\end{equation}
On the other hand, for fixed $ n $, if the transmission rate $ r $ increases, $ \Delta \rho^{\min} $ also increases and the gap to the normal approximation widens.
\end{remark}

\begin{lemma} \label{lemma_max_inf_rate_w_rate_gap}
	Lemma \ref{lemma_max_inf_rate} can also be introduced in terms of rate gap such as
	\begin{equation}\label{eq_min_req_rate_gap}
		M(n,\rho, \varepsilon) = R(n,\rho,\varepsilon) - \Delta r^{\min},
	\end{equation}
	where $ \Delta r^{\min} $ represents the minimum amount of rate penalty that is required to be paid to guarantee the desired CEP for fixed $ n $ and $ \rho $ under the latency, complexity, and reliability constraints.
\end{lemma}

\begin{remark}
	\eqref{eq_q_delta_tradeoff} is in accordance with the Properties listed in Sec. \ref{sec_dec_comp_prop}. For instance, using the statements in Properties 1 and 2, it is true that in order to achieve a desired CEP, it is possible to use a decoder with lower computational complexity as $ \Delta \rho $ increases. Similar results for $ \Delta r $ can also be deduced using Properties 3 and 4.
\end{remark}

\section{Optimal Rate-Power Selection}

The MOOP can now be reformulated for the OS decoder with the following additional constraints
\begingroup
\allowdisplaybreaks
\begin{subequations}
	\label{eq_mul_opt_prob_2}
	\begin{align}
	\underset{k,\varepsilon,s}{\text{minimize}} & ~~ [\Delta r, ~ \Delta \rho]
	\\
	\text{s.t.}  & ~~ \eqref{eq_mul_opt_problem_1_const_1}, \,\eqref{eq_mul_opt_problem_1_const_2}, \,\eqref{eq_mul_opt_problem_1_const_5},
	\\
	& ~~ 0\leq s\leq k, \label{eq_mul_opt_problem_2_const_6}
	\\
	& ~~ r_m \leq  r_s - \Delta r \leq r_s, \label{eq_mul_opt_problem_2_const_7}
	\\
	& ~~ \rho_s \leq \rho_s + \Delta \rho \leq \rho_m. \label{eq_mul_opt_problem_2_const_8}
	\end{align}
\end{subequations}
\endgroup
where \eqref{eq_mul_opt_problem_2_const_6} is the constraint on the order$ -s $. Additionally, we also introduce two constraints on rate and power. \eqref{eq_mul_opt_problem_2_const_7} represents the constraint on transmission rate selection, such that the rate cannot be below some value, denoted as $ r_m $. Similarly, \eqref{eq_mul_opt_problem_2_const_8} represents the constraint on power budget, such that the maximum power cannot be above some value, denoted as $ \rho_m $.

The MOOP introduced in \eqref{eq_mul_opt_prob_2} is a challenging problem due to the nonlinear relations between latency, reliability and power requirements. An exhaustive search over OS decoders with various orders at different rates is on the other hand a very complicated and inefficient solution to the problem. In this section, we investigate \eqref{eq_mul_opt_prob_2} and show the optimum solution. We first start with the following Lemma.

\begin{lemma}\label{lem_opt_w_equal_eps}
	The optimum is achieved with equality in \eqref{eq_mul_opt_problem_1_const_1}. 
\end{lemma}
\begin{proof}
	We prove the lemma by using a similar analogy from \cite[Lemma 6]{celebi_latency_and}. Without loss of generality let us first assume that $ \Delta r $ is selected to be the optimum, $ \Delta r  = \Delta r^*$, and we focus on the minimization of $ \Delta \rho $. For fixed $ \rho_s $, given that $ R(n,\rho_s,\varepsilon) \leq R(n,\rho_s,\varepsilon_m)$, the feasible set for $ n $ becomes the largest for $ \varepsilon=\varepsilon_m$. Also assume that the optimal decoder is $ D^*(n,r_s-\Delta r^*,\rho_s+\Delta \rho^*) $ with $\varepsilon^*<\varepsilon_m$. However, for some $ \Delta \rho^* \geq \sigma > 0 $ small enough, one can find a decoder $ D'(n,r_s-\Delta r^*,\rho_s+\Delta \rho^*-\sigma) $ that can achieve $ \varepsilon_m $ which requires lower SNR than the optimal one without violating the CEP constraint, which contradicts with the assumption.
	
	Similarly, now, we assume that $ \Delta \rho $ is optimum, $ \Delta \rho = \Delta \rho^* $, and the search is over $ \Delta r $. Suppose that the optimal decoder is $ D^*(n,r_s-\Delta r^*,\rho_s+\Delta \rho^*) $  with $\varepsilon^*<\varepsilon_m$. However, for some $ \alpha > 0 $ small enough, one can find a decoder $ D'(n,r_s-\Delta r^*+\alpha,\rho_s+\Delta \rho^*) $ with CEP $ \varepsilon_m $ that has lower rate gap, i.e. higher number of information bits transmitted with higher error rate but still does not violate the CEP constraint, which contradicts with the assumption. Hence, the optimum is achieved with equality in \eqref{eq_mul_opt_problem_1_const_1}.
\end{proof}

Next, we focus on the latency constraint, $ L_a \leq L_m $. Recall that this constraint is written in the form of a complexity bound on $ K(D) $ in \eqref{eq_upper_bound_on_K}. Using the model in \eqref{eq_q_delta_tradeoff}, this bound can be converted to a power penalty constraint using the the model proposed in \eqref{eq_q_delta_tradeoff}. Thus, in the light of Lemma \ref{lemma_max_inf_rate} and \ref{lemma_max_inf_rate_w_rate_gap}, the optimization problem now reduces to
\begin{subequations}
	\label{eq_mul_opt_prob_3}
	\begin{align}
	\underset{k, s}{\text{minimize}} & ~~ [\Delta r, ~ \Delta \rho]
	\\
	\text{s.t.}  & ~~ r_m \leq r_s - \Delta r \leq \min\big\{M(n,\rho_s + \Delta \rho, \varepsilon_m), r_s\big\} , \label{eq_mul_opt_problem_3_const_1}
	\\
	& ~~ 0 \leq \Delta \rho \leq \rho_m - \rho_s , \label{eq_mul_opt_problem_3_const_2}
	\\
	& ~~ 0 \leq s \leq s_m , \text{ and } \eqref{eq_mul_opt_problem_1_const_5}, \label{eq_mul_opt_problem_3_const_3}
	\end{align}
\end{subequations}
where the minimum in \eqref{eq_mul_opt_problem_3_const_1} is due to the fact that, by definition, $ \Delta r \geq 0 $. Notice that constraints on $ k $ and $ s $ can be omitted for further analysis, since their effect have already been represented in \eqref{eq_mul_opt_problem_3_const_1}. Thus, we have
\begin{subequations}
	\label{eq_mul_opt_prob_4}
	\begin{align}
	{\text{minimize}} & ~~ [\Delta r, ~ \Delta \rho]
	\\
	\text{s.t.}  & ~~  \eqref{eq_mul_opt_problem_3_const_1} \text{ and }  \eqref{eq_mul_opt_problem_3_const_2}.
	\end{align}
\end{subequations}
From \eqref{eq_mul_opt_prob_3}, the attainable objective set follows
\begin{equation}\label{key}
S = \Big\{ \{\Delta r, \Delta \rho\} \big| ~ \big[r_s -  M(n,\rho_s + \Delta \rho, \varepsilon_m)\big]^+  \leq \Delta r \leq r_s - r_m, 
\text{ for all } 0 \leq \Delta \rho \leq \rho_m - \rho_s \Big\} .
\end{equation}
An illustration on the attainable objective set $ S $ is depicted in Fig. \ref{fig_set_explanation}, where the reference transmission rate is $ r_s = 0.5  $ and constraints on rate, power, and latency are $ r_m = 0.25 $ and $\rho_m = 8 \, $dB, $ L_m = 1\,$ms, respectively. It can be seen that the attainable objective set $ S $ relies in between the constraints defined in \eqref{eq_mul_opt_problem_3_const_1}, \eqref{eq_mul_opt_problem_1_const_3}, and \eqref{eq_mul_opt_problem_1_const_4}. 

\begin{figure}[t]
	\centering
	\whencolumns{
		\includegraphics[width=.6\linewidth]{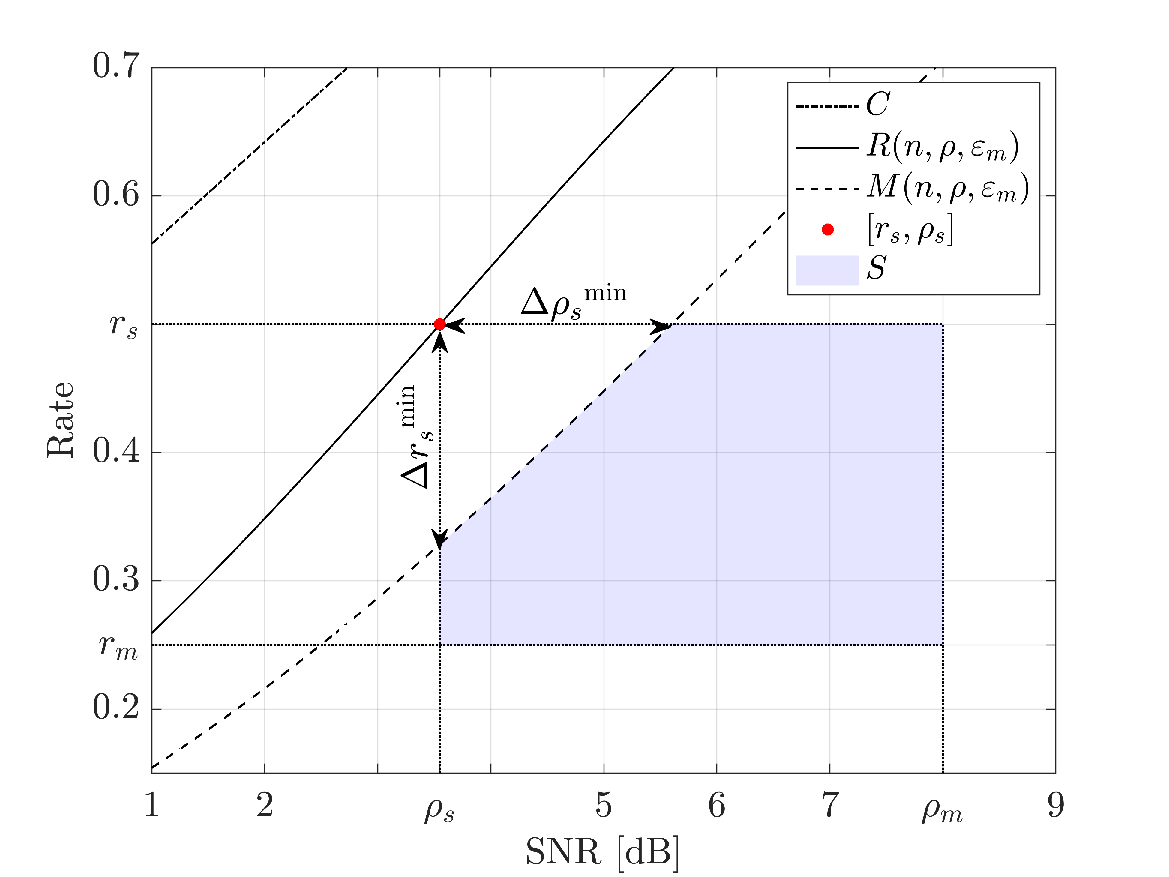}
	}{
		\includegraphics[width=1\linewidth]{figures/set_explanation.eps}
	}
	\caption{A numerical realization that shows the attainable objective set of the MOOP, denoted as $ S $, where the reference transmission rate is $ r_s = 0.5 $ and $ r_m = 0.25 $, $ \rho_m = 8 \,$dB. The ultimate transmission point $ [r_s, \rho_s] $ is shown with the star. As a comparison, the capacity and the maximum achievable rate, defined in \eqref{eq_normal_approximation}, are also depicted with the maximum achievable rate subject to latency, reliability, decoding complexity constraints, where $ n=128 $, $ \varepsilon_m = 10^{-5} $, $ L_m = 1\,$ms,  $ T_s = 1\,\mu$s, and $ T_b=1\,$ns.}
	\label{fig_set_explanation}
\end{figure}

\begin{remark}
	For a given $ \Delta\rho $, that is $ 0 \leq \Delta \rho \leq \rho_m - \rho_s  $, if $ M(n,\rho_s + \Delta \rho, \varepsilon_m) < r_m $, no feasible pair can be found.
\end{remark} 

\begin{lemma} \label{lemma_optimum_lies_on_M}
	The set of optimum solutions of \eqref{eq_mul_opt_prob} always leads to transmission pairs that lie on $ M(n,\rho,\varepsilon_m) $.
\end{lemma}
\begin{proof}
	 Recall that $ M(n,\rho,\varepsilon_m) $ represents the maximal achievable limit under latency, realibility, and decoding complexity constraints. Any transmission rate-power pair that violates $ M(n,\rho,\varepsilon_m) $, also violates the constraints in \eqref{eq_mul_opt_prob}. Thus, no better solution can be achieved above $ M(n,\rho,\varepsilon_m) $ and therefore the optimum selections of $ \Delta r $ and $ \Delta \rho $ yield the set of transmission pairs that lies on it.
\end{proof}

Lemma \ref{lemma_optimum_lies_on_M} shows that \eqref{eq_mul_opt_prob_3} is equivalent to the following
\begin{subequations}
	\label{eq_opt_problem_3}
	\begin{align}
	\text{minimize} & ~~ [\Delta r, ~ \Delta \rho]
	\\
	\text{s.t.}  
	& ~~ r_s - \Delta r = M(n,\rho_s + \Delta \rho, \varepsilon_m), \label{eq_mul_opt_problem_4_const_1}
	\\
	& ~~ \min\left\{\rho_s-\rho_m , \, \Delta {\rho_s}^{\min}\right\}  \geq \Delta \rho \geq 0, \label{eq_mul_opt_problem_4_const_2}
	\\
	& ~~  \min\left\{r_s-r_m , \, \Delta {r_s}^{\min}\right\} \geq \Delta r \geq 0 , \label{eq_mul_opt_problem_4_const_3}
	\end{align}
\end{subequations} 
where $ \Delta {\rho_s}^{\min} $ and $ \Delta {r_s}^{\min} $ represent the minimum amount of power penalty and rate gap that is required to meet the constraints at rate $ r_s $ and SNR $ \rho_s $, respectively. The minimum selections in \eqref{eq_mul_opt_problem_4_const_2} and \eqref{eq_mul_opt_problem_4_const_3} is due to the maximum limits of $ \Delta r $ and $ \Delta \rho $.

As an example, let the optimum $ \Delta r $ is selected to be $ 0 $, which leads to the rate $ r = r_s $. Assuming that $ \rho_s-\rho_m \geq \Delta {\rho_s}^{\min} $ , setting $ \Delta r = 0 $ only allows a horizontal movement in Fig. \ref{fig_set_explanation} and therefore a certain amount of power penalty needs to be added in order to meet the constraints. The optimum rate-power pair would be
\begin{equation}\label{eq_opt_trans_pair_alpha_1}
[ r_s, ~ \rho_s + \Delta {\rho_s}^{\min} ], 
\end{equation} 
since $ r_s = M(n,\rho_s + \Delta {\rho_s}^{\min},\varepsilon_m) $. 
Similar analysis can be applied if the optimum $ \Delta \rho $ is $ 0 $. Assuming that $ r_s-r_m \geq \Delta {r_s}^{\min} $, only a vertical movement is permitted. Since $ r_s - \Delta {r_s}^{\min} = M(n,\rho_s,\varepsilon_m) $, the optimum rate-power pair is 
\begin{equation}\label{eq_opt_trans_pair_alpha_0}
[ r_s - \Delta {r_s}^{\min}, ~ \rho_s ]. 
\end{equation} 

\begin{definition}[Pareto boundary \cite{cao_pareto_boundary}]
	The set of the optimal objectives of a MOOP is called the Pareto boundary. Any objective pair the belongs to the Pareto boundary cannot be dismissed since none of the objectives can be improved without degrading the others.
\end{definition}

\begin{corollary}
	$ M(n,\rho,\varepsilon_m) $ for $ \rho_s \leq \rho \leq \min\left\{\rho_m , \,\rho_s + \Delta {\rho_s}^{\min}\right\} $ represents the Pareto boundary.
\end{corollary}

\noindent Hence, any objective pair that belongs to the attainable objective region, $ S $, but not to the Pareto boundary is suboptimal, since there exist other operating points that are better or at least as good for every objective. Numerical realizations of Pareto boundaries that are achieved with different processor speeds, i.e. different $ T_b $ values, are depicted in Fig. \ref{fig_threeD_pareto_boundary}. Notice that the ultimate point $ [0,0] $ can be achieved when  $ T_b \leq 10^{-15}\,$s. However, as $ T_b $ increases, i.e. the processor speed decreases, the boundary moves away from the ultimate pair and a wider Pareto boundary is experienced, which leads to a broader set of optimal set of objectives. 

\begin{figure}[t]
	\centering
	\whencolumns{
		\includegraphics[width=.6\linewidth]{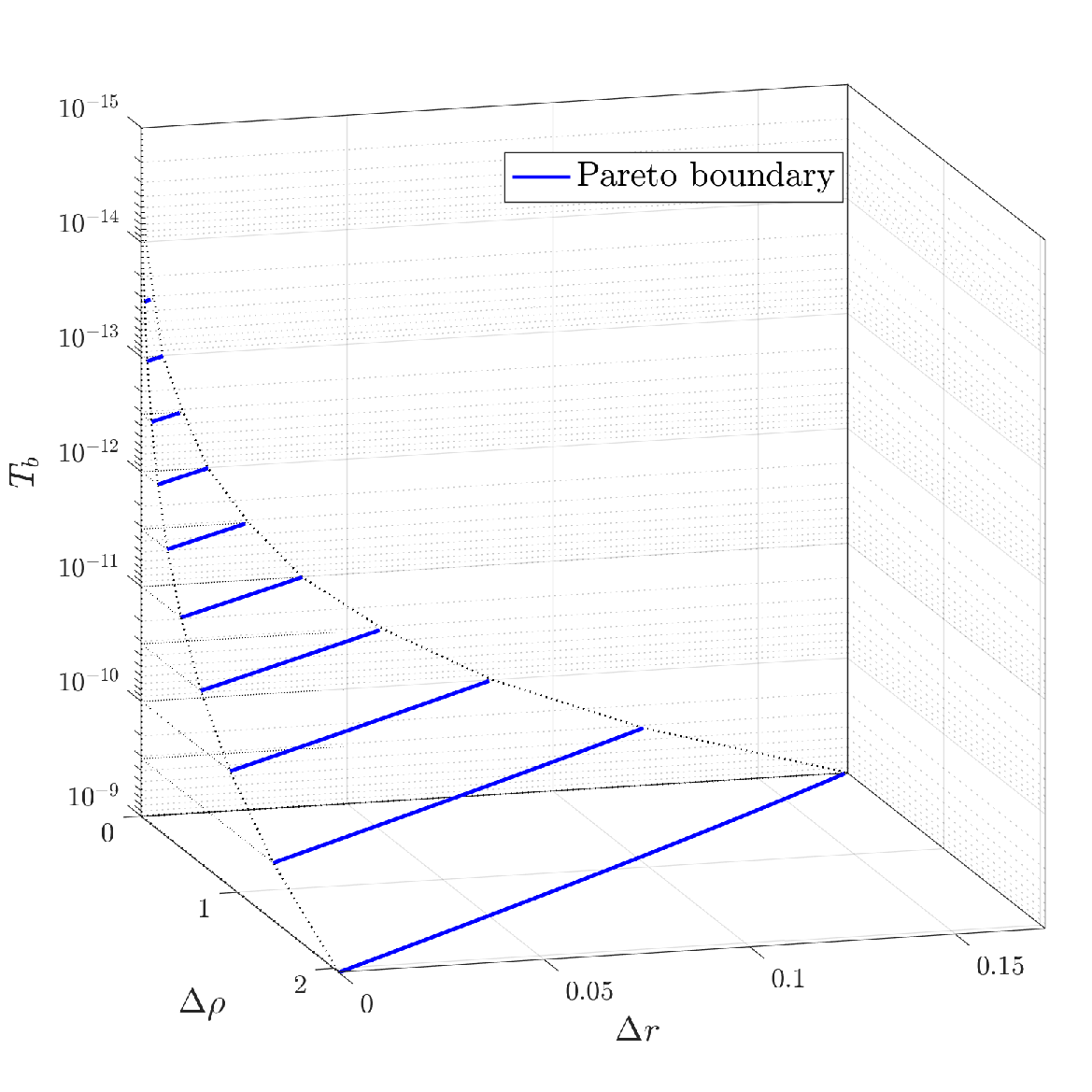}
	}{
		\includegraphics[width=1\linewidth]{figures/threeD_pareto_boundary.eps}
	}
	\caption{Numerical realizations of the Pareto boundaries with $ r_s = 0.5 $ for various $ T_b $, where $ n=128 $, $ \varepsilon_m = 10^{-5} $,  $ L_m = 1\,$ms, and $ T_s = 1\,\mu$s.}
	\label{fig_threeD_pareto_boundary}
\end{figure}

\section{Objective Scalarization}

Solution to the MOOP leads to a set of Pareto-optimal transmission pairs. When selecting a single optimum rather than a set of optimal points is desired, scalarization techniques are applied. The purpose of the scalarization is to aggregate the objectives into a single objective function and reduce the problem to a constrained single-objective optimization problem \cite{bjornson_multiobjective}. Several scalarization techniques have been presented in the literature, such as penalty-based intersection method, normal boundary intersection method, and weighted-$ l_\theta $ norm scalarization method \cite{taha_methods, chang_multiobjective, emmerich_a_tutorial}.  In this paper, the focus will be on weighted-$ l_\theta $ norm scalarization technique due to its simple structure and broad applicability. Mathematically, this technique can be applied to \eqref{eq_mul_opt_problem_1_obj} as \cite{taha_methods}
\begin{equation}\label{eq_general_obj_func}
\text{minimize} ~~ \left( A\alpha (\Delta r)^\theta + B(1-\alpha) (\log \Delta \rho)^\theta \right)^{\frac{1}{\theta}} ,
\end{equation} 
for $ \theta \geq 1 $ and $ 0 \leq \alpha \leq 1 $, where $ \alpha $ and $ (1-\alpha) $ represent the linear positive weight of the individual objectives and indicate the priority, $ \theta $ is the value of the norm, and finally $ A $ and $ B $ are some constant weights. For the simplicity of the analysis, we set $ A = 1 $ and $ B = 1 $. Observe that $ \Delta r $ and $ \Delta\rho $ do not share the same units. In the celebrated formula of Shannon for the real AWGN channel the rate is related to the signal power via the logarithmic function. For this reason, we also consider the logarithm of $ \Delta\rho $ in the objective function instead of $ \Delta\rho $ itself. 

Setting $ \alpha = 1 $, the objective function shifts to the minimization of $ \Delta r $. In this case $ \Delta r = 0 $ and the optimum transmission pair yields \eqref{eq_opt_trans_pair_alpha_1}. On the other hand, if $ \alpha = 0 $, the objective function considers only $ \Delta \rho $. Therefore, the optimum pair is the one shown in \eqref{eq_opt_trans_pair_alpha_0}. Pareto-optimal transmission pairs that lie in between \eqref{eq_opt_trans_pair_alpha_1} and \eqref{eq_opt_trans_pair_alpha_0} can be accessed by selecting different values of $ \alpha $, depending on the convexity-concavity of the Pareto boundary \cite{emmerich_a_tutorial}. 

The selection of the norm, $ \theta  $, is a significant determinant for accessibility of the Pareto-optimal pairs, due to the shape of the attainable objective set $ S $, which depends on the objective functions and constraints, some optimal pairs cannot be accessed with the selected scalarization function. This will be discussed further in the following Section. We continue our analysis by setting $ \theta = 1 $ and $ \theta = \infty $, which are the two most frequently used weighted-$ l_\theta $ norm scalarization techniques, namely linear weighted-sum and weighted Chebyshev, respectively.

\begin{figure*}[t]
	\begin{subfigure}{0.32\textwidth}
		\begin{subfigure}{1\textwidth}
			\centering
			\includegraphics[width=1\linewidth]{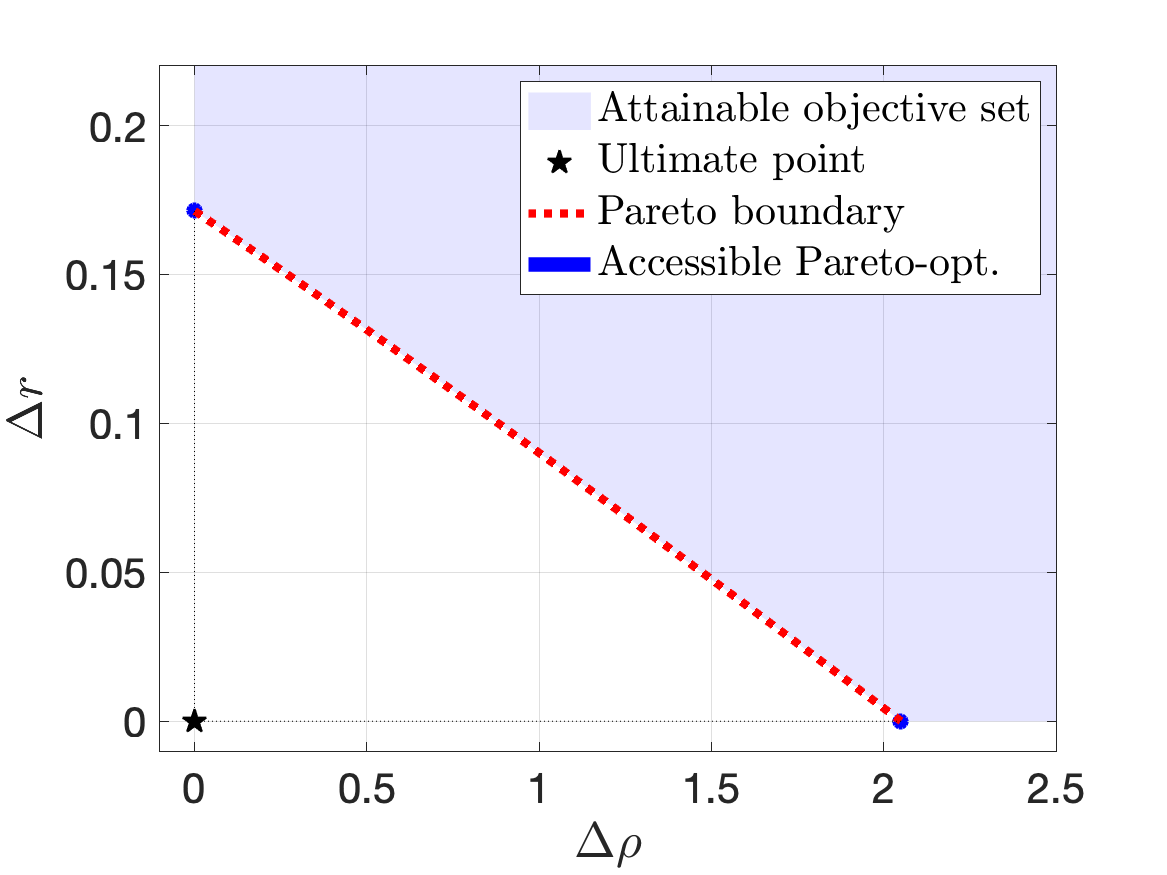}
		\end{subfigure}%
		\caption{$ r_s = 0.5$}
	\end{subfigure}
	\begin{subfigure}{0.32\textwidth}
		\begin{subfigure}{1\textwidth}
			\centering
			\includegraphics[width=1\linewidth]{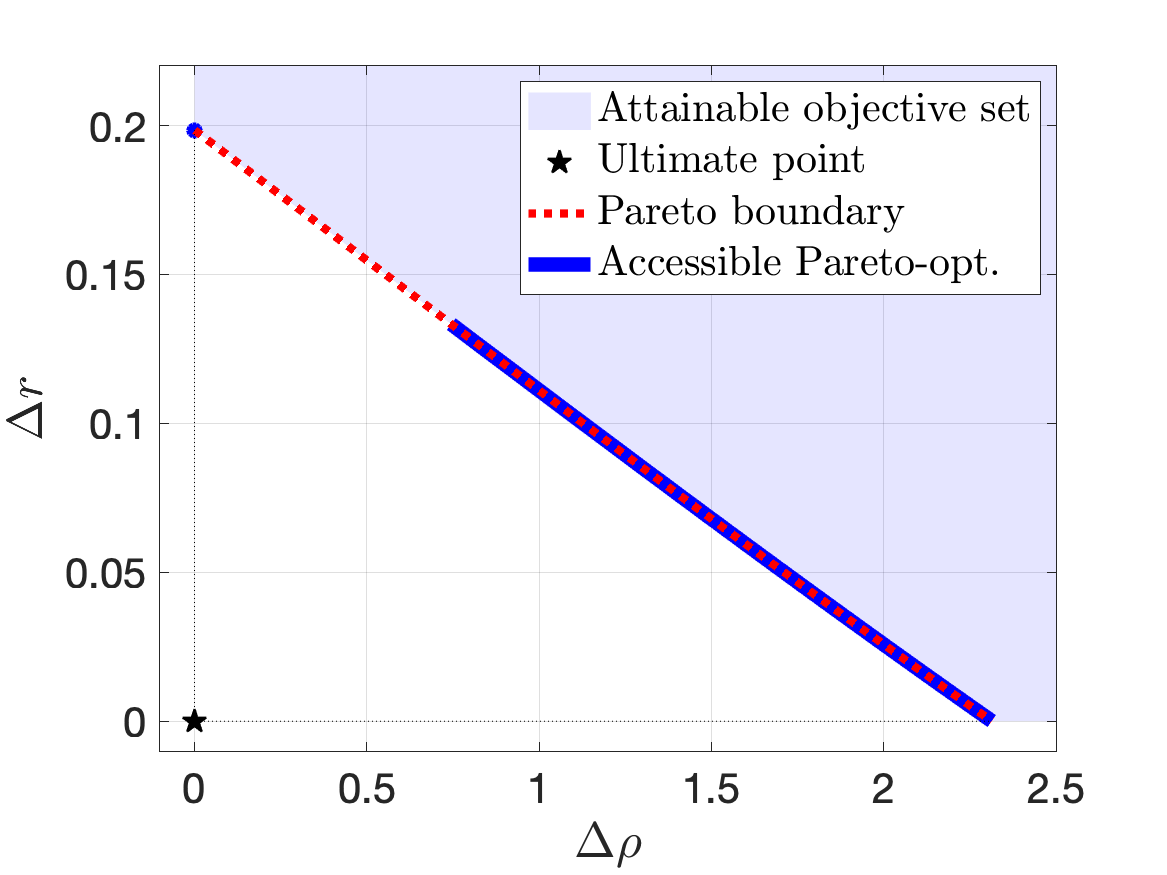}
		\end{subfigure}%
		\caption{$ r_s = 0.7$}
	\end{subfigure}
	\begin{subfigure}{0.32\textwidth}
		\begin{subfigure}{1\textwidth}
			\centering
			\includegraphics[width=1\linewidth]{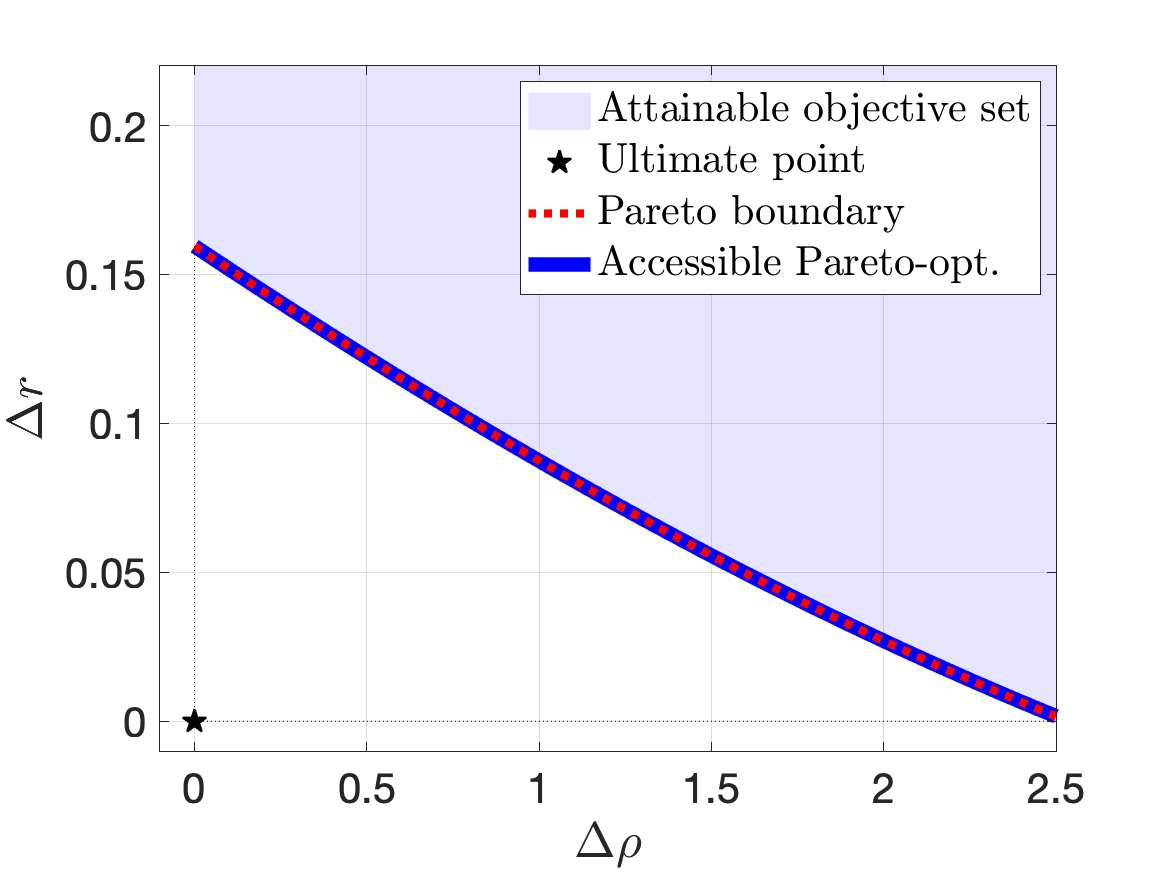}
		\end{subfigure}%
		\caption{$ r_s = 0.9$}
	\end{subfigure}
	\medskip
	
	\caption{Comparisons between the Pareto boundaries and the accessible Pareto-optimal points with the linear weighted-sum objective function  for $ r_s = \{0.5, 0.7, 0.9\} $ when $ n=128 $, $ \varepsilon_m = 10^{-5} $, $ T_s = 1 \,\mu$s, $ T_b = 1 \,$ns, and $ L_m = 1 \,$ms. Here, we assume that $ r_m $ and $ \rho_m $ are sufficiently low and high, respectively, such that they have no observable effect on the Pareto boundary.}
	\label{fig_pareto_all}
\end{figure*}

\subsection{Linear Weighted-Sum Objective Function: $ \theta = 1 $} \label{sec_linear_weighted_obj_func}

By setting $ \theta = 1 $, the objective function reduces to the weighted sum of the objectives. The objective function can be written as
\begin{equation}\label{eq_linear_weighted_sum_obj_func}
{\text{minimize}} ~~ \alpha \Delta r + (1-\alpha)\log\Delta \rho .
\end{equation}

The channel capacity $ C $ for BI-AWGN channel is bounded between $ [0,1] $. Due to its monotonicity in $ \rho $, it exhibits a sigmoidal shape. Similar to $ C $, \eqref{eq_normal_approximation} also follows a sigmoidal shape, which leads to having convex and concave portions that are separated with an inflection point. The inflection point represents the point where the shape of the maximal rate curve changes from convex to concave. It is denoted as $ \{r_i, \rho_i\} $ and is the solution to the equation 
\begin{equation}\label{key}
\frac{\partial^2 R(n,\rho,\varepsilon)}{\partial \rho^2} \Bigg|_{\rho = \rho_i} = 0.
\end{equation}
Since $  M(n,\rho,\varepsilon) $ is the $ \Delta {\rho}^{\min} $ amount shifted version of $  R(n,\rho,\varepsilon) $ and $ \Delta {\rho}^{\min} $ is monotonically increasing in $ \rho $, $  M(n,\rho,\varepsilon) $ also has a sigmoidal structure, where the inflection occurs at $ \rho_i + \Delta {\rho_i}^{\min} $.  

\begin{remark}\label{rem_regions}
	Due to the geometry of $  M(n,\rho,\varepsilon) $ in $ \rho $, the Pareto boundary is concave when $ \rho_s \leq \rho_i $ and convex when $ \rho_s > \rho_i + \Delta {\rho_i}^{\min} $. However, when $ \rho_i + \Delta {\rho_i}^{\min} \geq \rho_s > \rho_i $, the Pareto boundary consists of both concave and convex regions.
\end{remark}
Thus, depending on $ \rho_i $ and $ \rho_s $, the accessibility of the Pareto-optimal pairs with linear weighting objective function can be separated into three parts:

\subsubsection{Low-SNR regime ($ \rho_s \leq \rho_i $)} \label{sec_low_snr_region}

This is the regime where the Pareto boundary is concave. Only two Pareto-optimal points, which are the end points of the Pareto boundary, described in \eqref{eq_opt_trans_pair_alpha_1} and \eqref{eq_opt_trans_pair_alpha_0}, are accessible. 

\subsubsection{Medium-SNR regime ($ \rho_i + \Delta {\rho_i}^{\min} \geq \rho_s > \rho_i $)}

This is the regime where the Pareto boundary is convex over a certain region and concave over a certain other region, where these regions are specified in Remark \ref{rem_regions}. Within this regime, due to convexity, the optimal points located between $ \rho \in [\rho_i + \Delta {\rho_i}^{\min}, \rho_s + \Delta {\rho_s}^{\min}] $ are accessible. However, due to the concavity, optimal points between $ \rho \in (\rho_s, \rho_i + \Delta {\rho_i}^{\min}) $ are not accessible. Only the optimal point where $ \rho = \rho_s $ is accessible since it is the end point of the Pareto boundary. 

\subsubsection{High-SNR regime ($ \rho_s > \rho_i + \Delta {\rho_i}^{\min} $)} \label{sec_high_snr_region}

This is the regime where the Pareto boundary is convex and Pareto-optimal points in between \eqref{eq_opt_trans_pair_alpha_1} and  \eqref{eq_opt_trans_pair_alpha_0} are accessible.

This phenomenon can be seen in Fig. \ref{fig_pareto_all}.a, \ref{fig_pareto_all}.b, and  \ref{fig_pareto_all}.c, where Pareto boundaries and accessible Pareto-optimal $ \Delta r $ and $ \Delta \rho $ values are depicted for $ r_s = \{0.5, 0.7, 0.9\} $, respectively. When $ r_s = 0.5 $, the Pareto boundary is concave. Thus, only two Pareto-optimal points are accessible with the linear weighted-sum objective function. However, when $ r_s = 0.7 $, the Pareto boundary is convex over a part of the region and concave over the rest of the region and therefore the Pareto-optimal points on the convex part are accessible. When $ r_s = 0.9 $, the feasible region yields a convex Pareto boundary and enables the objective function to attain all Pareto-optimal points by selecting different values of $ \alpha $.

\subsection{Weighted Chebyshev Objective Function: $ \theta = \infty $}

Another way of scalarization can be achieved by setting $ \theta = \infty $ where now the objective function transforms to a weighted min-max formulation. This scalarization technique is known as the weighted Chebyshev objective function, which can be formulated as the following
\begin{equation}\label{eq_chebyshev_obj_func}
{\text{minimize}} ~~ \max\{ \alpha \Delta r, ~ (1-\alpha)\log\Delta \rho \}.
\end{equation}
The weighted Chebyshev objective function guarantees accessing all Pareto optimal points regardless of having concave or convex structure. 

\begin{remark}
	The Pareto boundary and its shape are not related to which scalarization method is selected. Its shape is drawn according to the set $ S $ which depends on the objective functions and constraints. The scalarization method only determines the accessible set of the Pareto optimal points. At the end, depending on $ A $, $ B $, and $ \alpha $, the scalarization function will find a final optimal point. Therefore, the final selected Pareto-optimal point will be different with different scalarization functions.
\end{remark}

\section{Case Study: Battery-Powered Transmission}

\begin{algorithm}[th]
	\caption{Multi-objective optimization}
	\begin{algorithmic}[1]
		\State Given $ n $ and $ \varepsilon_m $: \textbf{compute}: $ R(n,\rho,\varepsilon_m) $ using \eqref{eq_normal_approximation}
		\State Given $ 0 \leq r_s \leq 1 $: \textbf{compute}: $ \rho_s $ from \eqref{eq_reference_power}
		\State\textbf{compute}: $ \Delta {\rho}^{\min} $ using \eqref{eq_min_req_power_penalty}, $ \forall r \in (0,1] $
		\State\textbf{compute}: $ M(n,\rho,\varepsilon_m) $ using \eqref{eq_maximal_rate_with_latency}
		\State\textbf{compute}: $ \Delta r^{\min} $ using \eqref{eq_min_req_rate_gap}
		\If {$r_s - \Delta {r_s}^{\min} \geq r_m$} \State $ \Delta \rho^{\text{start}} = 0  $
		\Else \State {$ \Delta \rho^{\text{start}} = M^{-1}(n, r_m, \varepsilon_m)  $}
		\EndIf
		\If {$\rho_s + \Delta {\rho_s}^{\min} \leq \rho_m$} \State $ \Delta \rho^{\text{end}} = \Delta {\rho_s}^{\min}  $
		\Else \State {$ \Delta \rho^{\text{end}} = \rho_m - \rho_s  $}
		\EndIf
		\For {$\Delta \rho=\Delta \rho^{\text{start}}:\Delta \rho^{\text{end}}$}
		\State\textbf{compute}: $ \Delta r = r_s - M(n,\rho_s + \Delta \rho,\varepsilon_m) $
		\State Given $ t $: \textbf{compute}: $ \alpha $ from \eqref{eq_compute_alpha}
		\State Given $ A $, $ B $, and $ \theta $: \textbf{compute}: $ z(i) = \left( A\alpha (\Delta r)^\theta + B(1-\alpha) (\log \Delta \rho)^\theta \right)^{\frac{1}{\theta}} $
		\State $ i = i+1 $
		\EndFor
		\State Select $ [ r_s-\Delta r, ~ \rho_s+\Delta \rho ] $ minimizes $ z $ 
	\end{algorithmic} 
\end{algorithm}

In this section, we exemplify the importance of the MOOP by a case study. Suppose a battery-powered transmitter is communicating with a complexity constrained receiver, under latency and reliability constraints. The transmitter transmits codewords with $ n=128 $ blocklength under a latency constraint, defined in \eqref{eq_latency_constraint}. The objective is to maximize both the total number of information bits transmitted to the receiver until the battery dies and increase the energy efficiency at the same time. The transmitter and receiver set an ultimate transmission pair $ [r_s,\rho_s] $ and limits on rate and power, denoted as $ r_m $ and $ \rho_m $, respectively. The rate limit represents the minimum amount of information bits that must be sent with every codeword transmission and power limit is due to the power budget of the system.  Then, transceivers search for a transmission pair that can meet the constraints given the battery level. 

One can envision the communication environment as the following, when the battery is full, the transmitter can sacrifice power to maintain the constraints instead of introducing rate gap. Besides, as the remaining battery power level is decreasing, power is becoming more precious and the transmitter may need to introduce some amount of rate reduction instead of sacrificing more power. This can be achieved by defining the weight $ \alpha $ being related to the remaining battery power percentage, denoted as $ t $. Thus, a sigmoid relation between $ \alpha $ and $ t $, is defined as
\begin{equation}\label{eq_compute_alpha}
\alpha = 1-\left( 1 + \left( \frac{t}{1+t} \right)^2 \right)^{-1} .
\end{equation}
From \eqref{eq_compute_alpha}$, \alpha $ gets values close to $ 0 $ when the battery level is high and hence the goal of the optimization problem yields the minimization of $ \Delta r $. However, as the battery level decreases, $ \alpha $ increases and approaches $ 1 $. Therefore, the focus of the optimization problem shifts to the minimization of $ \Delta \rho $.\footnote{The relation between $ \alpha $ and $ t $ given in \eqref{eq_compute_alpha} can be accepted as a general example. Different relations according to the relevance of $ t $ in the case study can be proposed.} A significant observation here is that with the introduction of \eqref{eq_compute_alpha}, the transmission parameters are becoming adaptive with respect to $ t $ and constraints. 
An efficient algorithm to find the optimum transmission pair is shown in Algorithm 1.

\subsection{Numerical Results}

\begin{table}[t]
	\centering
	\caption{List of simulation parameters with their corresponding values.}
	\begin{tabular}{lll}
		\multicolumn{1}{c}{Parameter}         & \multicolumn{1}{c}{Value}        & \multicolumn{1}{c}{Comment}                        \\ \hline
		\multicolumn{1}{|l|}{$B$}             & \multicolumn{1}{l|}{$1$ Wh}      & \multicolumn{1}{l|}{Battery capacity}              \\ \hline
		\multicolumn{1}{|l|}{$r_s$}           & \multicolumn{1}{l|}{$0.5$}       & \multicolumn{1}{l|}{Target transmission rate}   \\ \hline
		\multicolumn{1}{|l|}{$r_m$}           & \multicolumn{1}{l|}{$0$}         & \multicolumn{1}{l|}{Minimum transmission rate}                  \\ \hline
		\multicolumn{1}{|l|}{$\rho_s$}        & \multicolumn{1}{l|}{$3.55$ dB}   & \multicolumn{1}{l|}{Target SNR}                 \\ \hline
		\multicolumn{1}{|l|}{$\rho_m$}        & \multicolumn{1}{l|}{$10$ dB}     & \multicolumn{1}{l|}{Maximum SNR}                   \\ \hline
		\multicolumn{1}{|l|}{$L_m$}           & \multicolumn{1}{l|}{$10^{-3}$ s}      & \multicolumn{1}{l|}{Maximum aggregate latency}     \\ \hline
		\multicolumn{1}{|l|}{$\varepsilon_m$} & \multicolumn{1}{l|}{$10^{-5}$}   & \multicolumn{1}{l|}{Maximum CEP}                   \\ \hline
		\multicolumn{1}{|l|}{$T_s$}           & \multicolumn{1}{l|}{$10^{-6}$ s} & \multicolumn{1}{l|}{Symbol duration}               \\ \hline
		\multicolumn{1}{|l|}{$T_b$}           & \multicolumn{1}{l|}{$10^{-9}$ s} & \multicolumn{1}{l|}{Processor speed}               \\ \hline
		\multicolumn{1}{|l|}{$p_a$}           & \multicolumn{1}{l|}{$30$ dB}     & \multicolumn{1}{l|}{Attenuation power per meter} \\ \hline
		\multicolumn{1}{|l|}{$d$}             & \multicolumn{1}{l|}{$100$ m}     & \multicolumn{1}{l|}{Distance between transceivers} \\ \hline
		\multicolumn{1}{|l|}{$\sigma^2$}      & \multicolumn{1}{l|}{-$110$ dBm}  & \multicolumn{1}{l|}{Noise power}                   \\ \hline
	\end{tabular}
	\label{tab_sim_param}
\end{table}

\begin{figure*}[t]
	\centering
	\begin{subfigure}{.49\textwidth}
		\centering
		\includegraphics[width=1\linewidth]{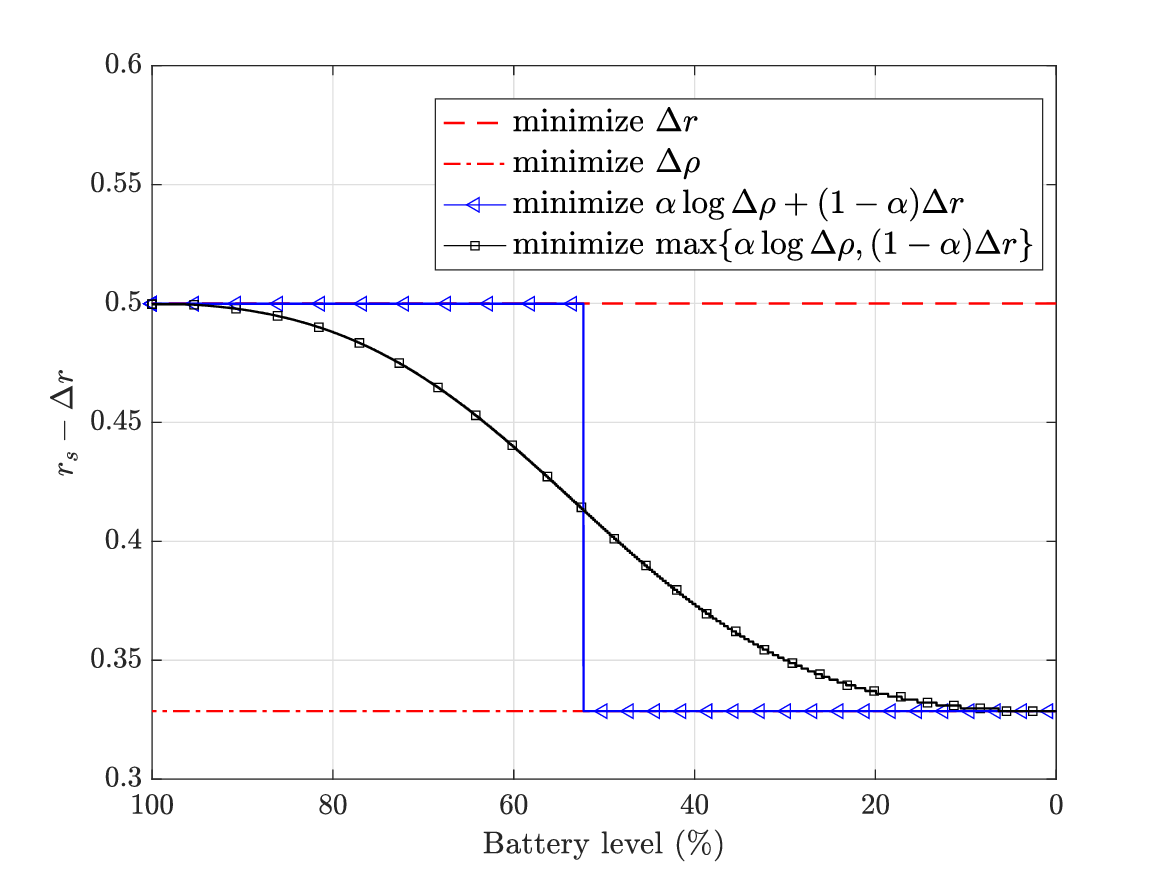}
		\caption{Comparison of the rate selections with respect to the remaining battery level.}
		\label{fig_batter_powered_single_scenario}
	\end{subfigure}~
	\begin{subfigure}{.49\textwidth}
		\centering
		\includegraphics[width=1\linewidth]{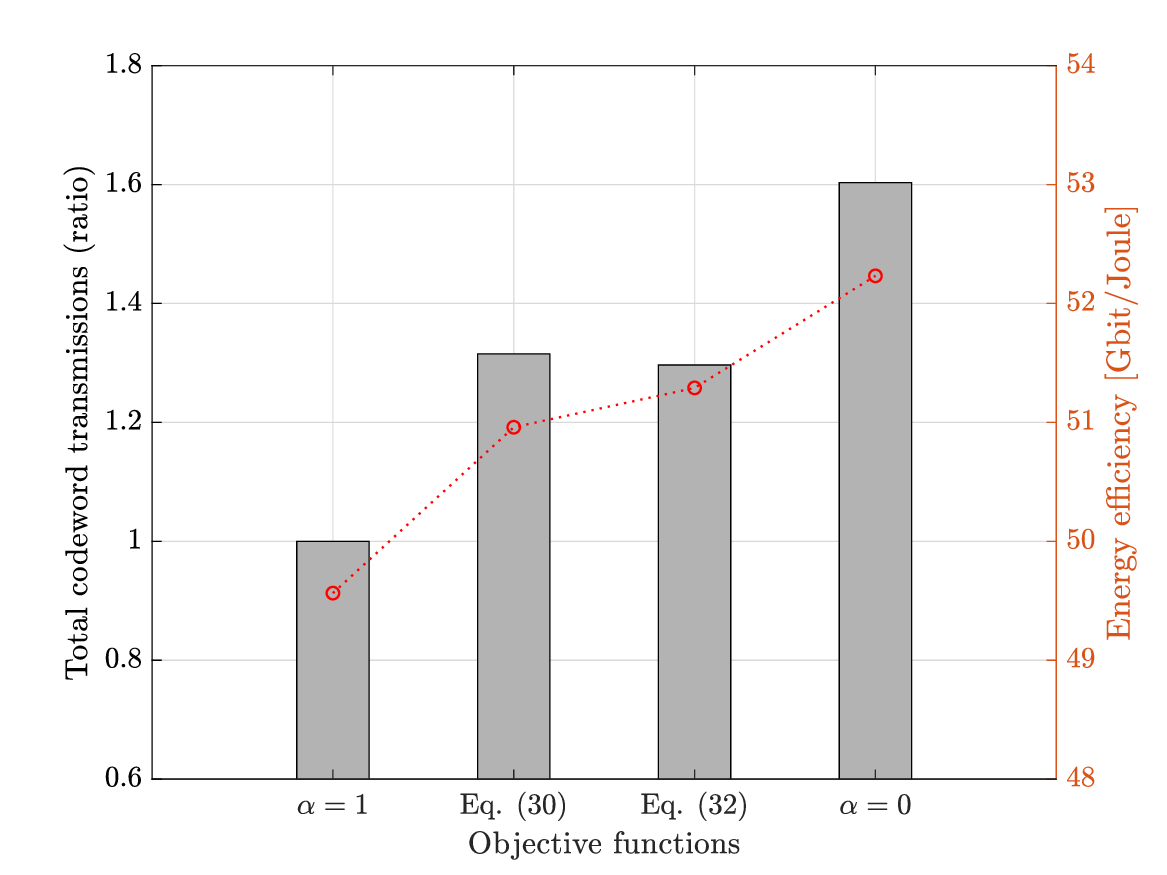}
		\caption{Ratio of the total number of codeword transmissions with different objective functions.}
		\label{fig_total_codeword_transmission}
	\end{subfigure}
	\caption{Comparison of the rate selections and total number of codeword transmissions for different objective functions, where $ r_s = 0.5 $: \textit{i}) minimization of $ \Delta r $: $ \alpha = 1 $, \textit{ii}) minimization of $ \Delta {\rho} $: $ \alpha=0 $, \textit{iii}) linear weighted-sum function, and \textit{iv}) weighted Chebyshev function.}
	\label{fig:test}
\end{figure*}

Suppose that a battery with $B = 1 \,$Wh capacity is used for data transmission at the transmitter. For each codeword transmission, the transmitter computes $ \alpha $ with respect to $ t $ and selects a transmission rate-power pair based on a target rate $ r_s $ and minimum rate $ r_m $ and the latency, reliability and decoding complexity constraints. Here, it is assumed that a sequence of data transmission is happening until the fully charged battery dies. Notice that as transmission continues, the transmitter calculates the remaining power from \cite{lauridsen_an}. As the battery power decreases, $ \alpha $ changes accordingly. Therefore, a new Pareto-optimal transmission pair is selected at each codeword transmission. This setup is an appropriate assumption for battery-powered URLLC IoT use cases where event-driven transmission is occurring \cite{feng_dynamic}. 

For the empirical analysis, we assume $30\,$dB average signal power attenuation at a reference distance of $1\,$m and $ -110 \,$dBm noise power at the receiver \cite{lopez_ultrareliable}. The distance between the transmitter and receiver is set to $ 100 \,$m. The list of all simulation parameters with their corresponding values is shown in Table \ref{tab_sim_param}.

We first set $ r_s = 0.5 $ and $ r_m = 0 $. Numerical results on the selected Pareto-optimal transmission rate at a codeword transmission with respect to $ t $ for linear weighted-sum and weighted Chebyshev functions are depicted in Fig. \ref{fig_batter_powered_single_scenario}. Additionally, for comparison purposes, we introduce two new objective functions: \textit{i}) minimize $ \Delta r $, i.e. $ \alpha=1 $, \textit{ii})  minimize $ \Delta \rho $, i.e. $ \alpha=0 $, which are shown with red dashed and dash-dotted lines, respectively. These objective functions do not depend on $ t $ and can be assumed as examples for fixed parameter transmission.

It can be seen that setting $ \alpha = 1 $ and $ \alpha = 0 $ yield two distinct choices, i.e. $ r_s $ and $ r_s - \Delta {r_s}^{\min} $ for all $ \alpha $. More interesting pictures arise, as we allow $ \alpha $ to vary with $ t $ and select other objectives. We first start with the linear weighted-sum objective function. As expected, due to the concave structure of the attainable objective set, only two Pareto-optimal points are accessible
\begin{equation}
	r = \begin{cases}
	r_s & \text{when $t$ is high}\\
	r_s - \Delta {r_s}^{\min} & \text{when $t$ is low} .
	\end{cases}
\end{equation}
Notice from  Fig. \ref{fig_batter_powered_single_scenario} that a transition from $ r_s $ to $ r_s - \Delta {r_s}^{\min} $ happens approximately when the remaining battery level falls below 50$ \% $. Besides, unlike the linear weighted-sum, optimum rate selection of the weighted Chebyshev function follows a smooth transition from $ r_s $ to $ r_s - \Delta {r_s}^{\min} $ and access all Pareto-optimal points in-between. 

Energy efficiency is crucial for IoT setups since very long battery life is required \cite{lauridsen_an}. Numerical results in Fig. \ref{fig_batter_powered_single_scenario} show that the selected Pareto-optimal transmission rate decreases as $ t $ reduces. Although this decreases the throughput of the communication system, it also allows to select lower transmit power, which, in the long run, allow more codeword transmissions and increase the energy efficiency in terms of total number of transmitted information bits per Joule. To see this effect, total number of codeword transmissions with a fully charged battery are depicted in Fig. \ref{fig_total_codeword_transmission}, where the left vertical axis represents the ratio of total number of codeword transmissions to the case where $ \alpha=0 $ and the right vertical axis shows the energy efficiency which is computed according to the following formula \cite{bjornson_how}
\begin{equation}\label{key}
\text{Energy efficiency} = \frac{n \sum (r_s - \Delta r)}{3600 B} ~~ [\text{bits}/\text{Joule}] ,
\end{equation}
where the sum is over all codeword transmissions until the battery dies. Results show that by implementing the MOOP, it is possible to transmit approximately 30\% more codewords and increase the power efficiency when linear weighted-sum or weighted Chebyshev objective functions are selected. 


\begin{figure}[t]
	\centering
	\includegraphics[width=1\linewidth]{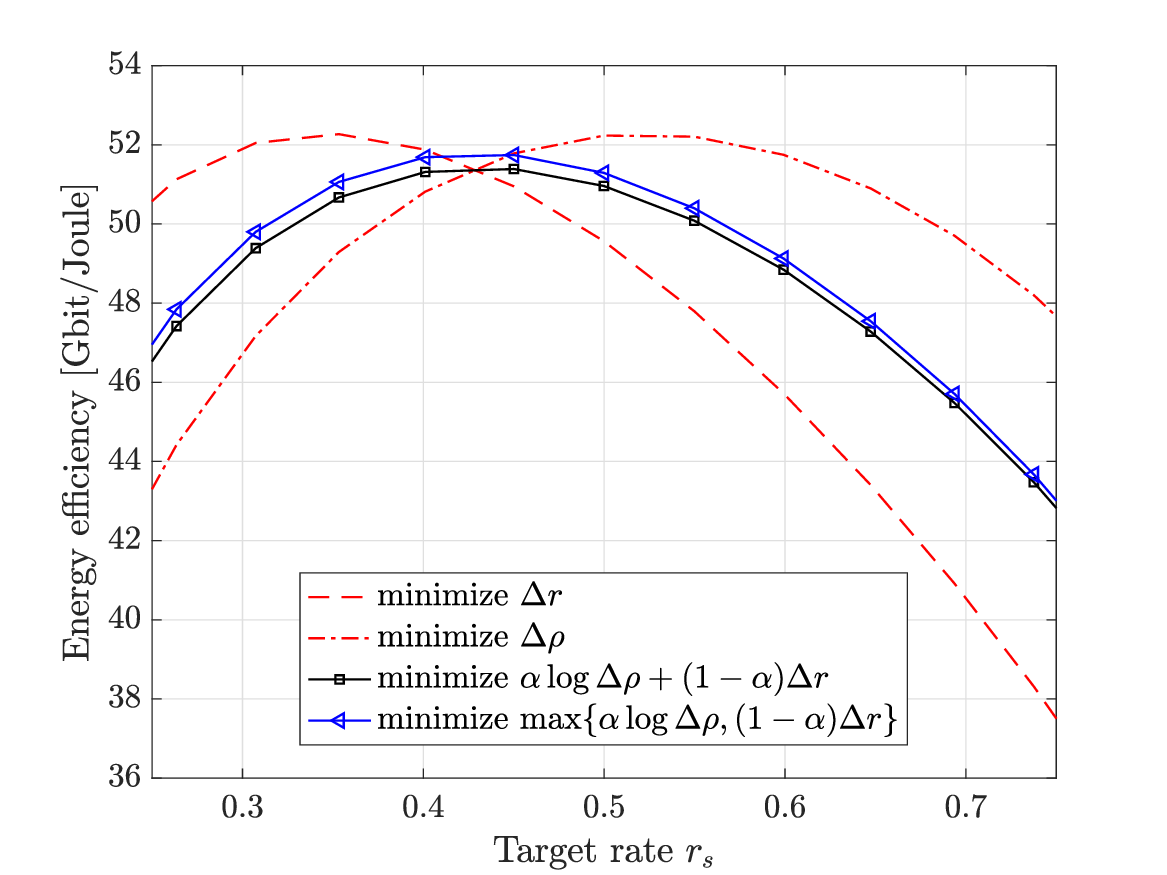}
	\caption{Comparison of the objective functions in terms of energy efficiency with respect to various target rates: \textit{i}) minimization of $ \Delta r $: $ \alpha = 1 $, \textit{ii}) minimization of $ \Delta {\rho} $: $ \alpha=0 $, \textit{iii}) linear weighted-sum function, and \textit{iv}) weighted Chebyshev function.}
	\label{fig_batter_powered_scenario}
\end{figure}

Next, we investigate the energy efficiency values of the MOOP framework for various target rates. Based on the results, which are depicted in Fig. \ref{fig_batter_powered_scenario}, one can divide the figure into two regions: low and high target rate regions. In low target rate region, 
the objective function with $ \alpha=0 $ achieves higher energy efficiency compared to the other objectives. On the contrary, in high rate region, 
the maximum energy efficiency values are achieved when $ \alpha=1 $. However, notice that for both regions the linear weighted-sum and weighted Chebyshev functions achieve a performance that is in between. 
Also notice that although the weighted Chebyshev function yields less number of codeword transmissions compared to linear weighted-sum function, as shown in Fig. \ref{fig_total_codeword_transmission}, it leads to higher energy efficiency for all $ r_s $. Finally, at the intersection of the low and high rate regions, the weighted Chebyshev function outperforms all other objective functions and achieves the highest energy efficient communication.  

\section{Conclusions and Future Directions}

We study the optimum transmission parameters for linear block codes with OS decoders in finite blocklength under latency, reliability, and decoding complexity constraints. It is shown that constraints introduce a back-off from the finite maximal achievable rate. Based on this analysis, a MOOP is formulated and solved with the help of an empirical model that can accurately track the trade-off between power gap versus per-information-bit computational complexity of the decoder. The attainable objective set and the Pareto boundary are characterized for two different multi-objective scalarization functions, namely the linear weighted-sum and weighted Chebyshev functions. Then, the accessability of the Pareto-optimal points are analyzed. It is also shown that depending on the convexity/concavity of the Pareto boundary, the Pareto-optimal points cannot be accessed by the linear weighted-sum objective function. Finally, benefits of the MOOP are shown with a case study on battery-powered transmission, where the weights for the scalarization of the MOOP is configured according to the remaining battery level and thus the optimum selection of the MOOP varies with the remaining battery level. This, therefore, introduces an adaptive transmission parameter selection. It is further shown that, MOOP framework increases both the throughput and energy efficiency of the system, compared to the classical fixed parameter transmission, while the constraints on latency, reliability, and decoding complexity are still met. Finally, application of the proposed MOOP framework on quasi-static fading channels is left as a potential future work.

\bibliographystyle{IEEEtran}
\bibliography{low_latency_multi_obj}

\balance

\end{document}